\documentclass[12pt]{article}
\usepackage[utf8x]{inputenc}
     
\usepackage{a4wide}
\usepackage{amssymb}
\usepackage{amsmath}
\usepackage{amsthm}
\usepackage[mathscr]{euscript}
\usepackage[colorlinks=true]{hyperref}
\usepackage{authblk}

\usepackage[normalem]{ulem}

\numberwithin{equation}{section}

\theoremstyle{definition}
\newtheorem{definition}{Definition}[section]

\theoremstyle{plain}
\newtheorem{Theorem}[definition]{Theorem}
\newtheorem{Proposition}[definition]{Proposition}
\newtheorem{Lemma}[definition]{Lemma}
\newtheorem{Corollary}[definition]{Corollary}

\theoremstyle{remark}
\newtheorem{remark}[definition]{Remark}

\newcommand\red[1]{{\color{red} {#1}}}

\newcommand{\R}{\mathbb R}

\newcommand{\eps}{\varepsilon}

\usepackage[xcolor]{changebar}
\cbcolor{blue}
\newcommand{\Ric}{\mathrm{Ric}}

\newcommand{\sse}{\subseteq}

\newcommand{\go}{\mathring{g}}
\newcommand{\Mo}{\mathring{M}}
\newcommand{\fo}{\mathring{f}}
\usepackage[inline]{enumitem}
\newcommand{\enumlabelformat}{\roman}
\newcommand{\enumlabelfont}[1]{#1}
\newlength{\thelabelsep}
\setlist{labelsep=\thelabelsep}
\setlist[enumerate]{font=\enumlabelfont,label=(\enumlabelformat*),leftmargin=2.5em}
\setlist[itemize]{leftmargin=2.5em,label=$-$}

\newcounter{inlineenum}
\renewcommand{\theinlineenum}{\enumlabelformat{inlineenum}}

\newcommand\beq{\begin{equation}}
\newcommand\eeq{\end{equation}}
\newcommand\ben{\begin{enumerate}}
\newcommand\een{\end{enumerate}}
\newcommand\bit{\begin{itemize}}
\newcommand\eit{\end{itemize}}
\newcommand{\Io}{\mathring{I}}

\renewcommand\a{\alpha}

\newcommand\s{\sigma}
\renewcommand\d{\partial}
\newcommand\e{\epsilon}

\newcommand\g{\gamma}

\newcounter{mnotecount}

\title{Rigidity of asymptotically  $AdS_2 \times S^2$ spacetimes}

\author[*]{Gregory J. Galloway}
\author[$\dag$]{Melanie Graf}

\affil[*]{Department of Mathematics, 

University of Miami, Coral Gables, FL, USA}

\affil[$\dag$]{Faculty of Mathematics,

University of Vienna, Vienna, Austria}

\begin{document}

\date{\today}

\maketitle

\begin{abstract}

The spacetime $AdS_2 \times S^2$ is well known to arise as the `near horizon' geometry of the extremal Reissner-Nordstrom solution, and for that reason it has been studied in connection with the AdS/CFT correspondence.  Here we consider asymptotically $AdS_2 \times S^2$ spacetimes that obey the null energy condition (or a certain averaged version thereof).
Supporting a conjectural viewpoint of Juan Maldacena, we show that any such spacetime must have a special geometry similar in various respects to $AdS_2 \times S^2$, and under certain circumstances must be isometric to  $AdS_2 \times S^2$.

\end{abstract}

\section{Introduction}

An interesting feature of the spacetime $AdS_2 \times S^2$ is that it arises as the `near horizon' geometry of the extremal Reissner-Nordstrom solution; see e.g. \cite{KL13}.  For this reason, this spacetime (sometimes referred to as the Robinson-Bertotti solution) has been studied in various works in connection with the AdS/CFT correspondence \cite{Mald98}; see e.g. \cite{MMS} and references therein. More recently a class of horizon free supersymmetric solutions to Einstein-Maxwell theory having $AdS_2 \times S^2$ asymptotics has been constructed by Lunin \cite{Lunin}. 
However, on the basis of an example considered in \cite[Section 2.2]{MMS}, and also a result in \cite{GSWW} (Theorem 2.1), Maldacena has suggested that any asymptotically 
$AdS_2 \times S^2$ spacetime that obeys the null energy condition (NEC), or more generally the average null energy condition (ANEC), should be quite special.  
In fact he has suggested the conjecture that any such spacetime should be isometric to $AdS_2  \times S^2$ \cite{MaldPC}.
In particular, consistent with the example in \cite{MMS} mentioned above, $4$-dimensional spacetimes that satisfy the ANEC {\it strictly} could not have $AdS_2 \times S^2$ asymptotics.   All of this suggests that the examples constructed in \cite{Lunin} cannot be globally regular, which, in fact, has since been confirmed by Lunin \cite{MaldPC}.

In this paper we obtain some results on the {\it rigidity} of asymptotically $AdS_2 \times S^2$ spacetimes satisfying the NEC, which support the conjectural picture put forth by Maldacena. While precise statements are postponed to Section \ref{sec:main}, our main result may be paraphrased as follows.

\begin{Theorem}\label{thm:main}
Let $(M,g)$ be an asymptotically $AdS_2 \times S^2$ spacetime (see Definition \ref{def:assumptions}) that satisfies the null energy condition (NEC), $\Ric(X,X) \ge 0$ for all null vectors $X$. Then the following holds.
\ben
\item $(M,g)$ is foliated by smooth totally geodesic null hypersurfaces $N_u \approx \R \times S^2$, $u \in \R$.
\item By time-dualizing, one obtains a second foliation by smooth totally geodesic null hypersurfaces $\hat{N}_v \approx \R \times S^2$, $v \in \R$, transverse to the foliation  $\{N_u\}_{u \in \R}$.  By considering the intersections of the $N_u$'s and $\hat{N}_v$'s, this double null foliation gives rise  to a foliation of $(M,g)$ by totally geodesic isometric round (i.e. constant curvature) 2-spheres S(u,v).  
\een
\end{Theorem}

\medskip
The properties (i) and (ii) are, of course, basic features of $AdS_2 \times S^2$.  One of the main results leading to the proof of Theorem \ref{thm:main}, Proposition  \ref{prop:existenceofnulllines},
together with a known result concerning the existence of null conjugate points \cite{Tipler78, EK94} confirms that there do not exist any asymptotically $AdS_2 \times S^2$ spacetimes obeying the strict ANEC.

While we have stated Theorem \ref{thm:main}  and Theorem \ref{thm:main2} below
with respect to the NEC, in fact both results remain valid under a weaker curvature condition (which, however, is stronger than the ANEC): It is sufficient to assume that along all future or past complete {\it null rays} $\eta: [0,\infty) \to M$, one has,
\beq\label{eq:averic}
\int_0^{\infty} \Ric(\eta'(s),\eta'(s)) ds \ge 0  \,.
\eeq
In order to simplify a bit the presentation of the proofs of Theorems \ref{thm:main} and \ref{thm:main2}, we postpone to an appendix a discussion of the changes needed to prove these theorems under the curvature condition \eqref{eq:averic}.

Theorem \ref{thm:main}  falls short of showing that $(M,g)$ splits as a metric product along the totally geodesic $2$-spheres.  A necessary condition for this is that the distribution of timelike $2$-planes orthogonal to the $2$-spheres be integrable.  
However, a recent example of Paul Tod \cite{Tod} shows that even if $(M,g)$ did admit such a product structure, it need not be isometric to $AdS_2 \times S^2$.
In general, some additional condition is needed to show that $(M,g)$ is isometric to $AdS_2 \times S^2$.  Such a condition is considered in the next result.
 
 Although  not itself an Einstein manifold, $AdS_2 \times S^2$ is a product of Einstein manifolds, and as such its Ricci tensor is covariant constant, $\nabla \Ric = 0$.  Under this added assumption we obtain the following.

\begin{Theorem}\label{thm:main2}
Let $(M,g)$ be an asymptotically $AdS_2 \times S^2$ spacetime that satisfies the NEC.  If the Ricci tensor is covariant constant then $(M,g)$ is globally isometric to $AdS_2 \times S^2$.
\end{Theorem}

\noindent
{\it Remark:}  $AdS_2 \times S^2$ has vanishing scalar curvature, $R = 0$. If one could establish a local metric splitting along the totally geodesic 2-spheres, then adding this condition to the assumptions of Theorem \ref{thm:main}, would be sufficient  to give that $(M,g)$ is isometric to $AdS_2 \times S^2$ (cf. section 4).  This condition, in particular, rules out examples like that of Tod.

\medskip
We would like to say a word about the approach to the asymptotics taken here.  One possible approach to the asymptotics, which will be considered in a subsequent paper, is to introduce a notion of a `singular' timelike conformal boundary.  In fact, $AdS_2 \times S^2$ admits, in a fairly natural way, such a boundary.   The more customary analytic approach taken in the present paper, is to require that the spacetime metric $g$ asymptote at a suitable rate, with respect to a natural coordinate system, to the $AdS_2 \times S^2$ metric $\go$ on approach to infinity.  This approach to the asymptotics gives strong control over the causal structure and allows one to obtain  rather fine geometric properties needed to establish Theorem \ref{thm:main}.

In Section 2 we give the formal definition of an asymptotically $AdS_2 \times S^2$ spacetime, and derive some consequences of the assumed asymptotics.  In Section 3 we establish the existence of a foliation by totally geodesic null  hypersurfaces, and a foliation by totally geodesic isometric round $2$-spheres, thereby establishing Theorem~\ref{thm:main}.  In Section 4 we present a proof of Theorem \ref{thm:main2}. 

\medskip

\noindent
\textsc{Acknowledgements.}  The authors are very grateful to Juan Maldacena for bringing this problem to their attention and for many valuable comments. The authors would also like to thank Eric Ling for his interest in this work and for helpful comments. GJG's research was partially supported by the NSF under the grant DMS-171080. MG's research was supported by project P28770 of the Austrian Science Fund FWF and a scholarship from the Austrian Marshall Plan Foundation to visit the University of Miami.

\section{Asymptotically $AdS_2 \times S^2$ spacetimes}

In this section we describe in a precise manner what it means for a spacetime $(M,g)$ to be asymptotically $AdS_2 \times S^2$, and we obtain some consequences of these assumed  asymptotics.

\subsection{Exact $AdS_2 \times S^2$ space} \label{sec:go}

Let $\Mo=\R \times \R \times S^2$. We set 
\[\go =-\cosh^2(x) dt^2 +dx^2 +d\Omega^2.\] 

For future reference, the non-zero Christoffels for this metric are
\[\mathring{\Gamma}^t_{tx}=\tanh(x),\, \mathring{\Gamma}^x_{tt}=\cosh(x)\sinh(x),\, \mathring{\Gamma}^\phi_{\phi\theta}=\cot(\theta),\, \mathring{\Gamma}^\theta_{\phi\phi}=\sin(\theta)\cos(\theta),\] the Riemann tensor can be expressed as
\[\mathring{R}=\mathring{R}_{AdS_2}+\mathring{R}_{S^2}\]
and the same holds for $\mathring{\Ric}$. Explicitly one has
\[\mathring{\Ric}_{tt}=\cosh(x)^2,\, \mathring{\Ric}_{xx}=-1,\, \mathring{\Ric}_{\theta\theta}=1,\, \mathring{\Ric}_{\phi\phi}=\sin^2(\theta).\]The scalar curvature vanishes.
Note that while $AdS_2\times S^2$ is not an Einstein manifold, one can still nicely express $\mathring{\Ric}$ in terms of the metric: $\mathring{\Ric}=-\go+2d\Omega^2$.

\subsection{The metrics $\go_\alpha$}
To get a better handle on the asymptotics of $g$ we will further define a family of metrics $\go_\alpha$ ($\alpha \in \R_+$) on $\Mo $ via
\[\go_\alpha =-\alpha \cosh^2(x) dt^2 +dx^2 +d\Omega^2. \]
The importance of these metrics for the asymptotics lies in Lemma \ref{lem:betaalpha}, stating that, in essence, there exist $\beta >1$ and $\alpha <1$ such that far out $\go_{\alpha}\prec g \prec \go_{\beta}$ and $\beta, \alpha \to 1$ as one approaches infinity.
(Recall, for Lorentzian metrics $g_1$ and $g_2$, $g_1 \prec g_2$ means that the null cones of $g_2$ are wider than those of $g_1$ in the sense that for any vector $X \ne 0$, if $g_1(X,X) \le 0$ then $g_2(X,X) < 0$).

Along the null curves of $\go_\a$ with $\theta = \theta_0$, $\phi = \phi_0$, one has
$$
dt = \pm \frac1{\sqrt{\a} \cosh x} dx \,. 
$$
Integrating gives the following. 

\begin{Lemma}[Null curves for $\go_\alpha$]\label{lem:nulllinseforgbeta}
	The curves $s\mapsto (f_\alpha(s,t_0,x_0), s+x_0,\phi_0,\theta_0)$ and $s\mapsto (-f_\alpha(s,-t_0,x_0), s+x_0,\phi_0,\theta_0)$, where
$$
f_\alpha(s,t_0,x_0):=\frac{2}{\sqrt{\alpha}} \big( \tan^{-1}(e^{s+x_0})-\tan^{-1}(e^{x_0}) \big) +t_0 \,,
$$
are future, resp.\ past, directed achronal null curves in $(\Mo, \go_{\alpha})$ passing through the point $(t_0,x_0,\phi_0,\theta_0)$.
\end{Lemma}

\begin{remark}\label{rem:behaviouroff}
For future reference we note the following.
\[\lim_{s\to-\infty} f_\alpha(s,t_0,x_0)=t_0-\frac{2 \tan^{-1}(e^{x_0})}{\sqrt{\alpha}}\] and \[\lim_{s\to \infty} f_\alpha(s,t_0,x_0)=t_0-\frac{2 \tan^{-1}(e^{x_0})}{\sqrt{\alpha}}+\frac{\pi}{\sqrt{\alpha}}\leq t_0+\frac{\pi}{\sqrt{\alpha}}.\] 
Since $f_\alpha$ is increasing this means that $f_\alpha(s,t_0,x_0)\leq t_0-\frac{2 \tan^{-1}(e^{x_0})}{\sqrt{\alpha}}+\frac{\pi}{\sqrt{\alpha}}$ for all $s$.
\end{remark}

The timelike futures $\Io^+_\alpha(p)$ are easily seen to satisfy the following.

\begin{Lemma}\label{lem:t+piinIp}
	Let $p=(t_0,x_0,\omega_0)\in (\Mo, \go_{\alpha})$. Then $\{t> t_0+\frac{\pi}{\sqrt{\alpha}}\}\times \R \times \{\omega_0 \} \sse \Io^+_\alpha(p)$.
\end{Lemma}

\subsection{Definition of asymptotically $AdS_2\times S^2$ spacetimes}\label{sec:definition}

Throughout we shall assume that spacetime is {\it causally simple}.  Following \cite[Sec.~3.10]{MS08}, we say that a spacetime is causally simple provided $J^\pm(p)$ is closed for all $p\in M$ and $(M,g)$ is causal (i.e.\ contains no closed causal curves).  As a consequence, the sets $J^\pm(K)$ are closed for all compact sets $K$ in $M$,  and $(M,g)$ is strongly causal.


In order to prove our main results, a careful treatment of the asymptotics, as layed out in the following definition,  is required.

\begin{definition}\label{def:assumptions}  Lat $(M,g)$ be a $4$-dimensional causally simple spacetime. 
We say that $(M,g)$ is asymptotically $AdS_2\times S^2$ provided the following conditions hold.	
\ben
\item[($a_1$)] There exists a closed subset $A\sse M$ such that $M\setminus A^\circ $ is the disjoint union of two manifolds with boundary $M_1$ and $M_2$ such that 
\begin{align*}
M_1 \cong \R \times (-\infty,-a] \times S^2 \;\; &\mathrm{and} \;\; M_2 \cong \R \times [a,\infty ) \times S^2, \\
p  \in M_1 \cup M_2 &\mapsto (t(p),x(p), \omega(p)) \,,   
\end{align*}
$a\geq 1$,  and the boundary $\partial A$ is mapped to $(\R \times \{-a\}\times S^2) \cup (\R \times \{a\}\times S^2)$.   

\item[($a_2$)] For all $p\in A$ and  $k=1, 2$:
		\[I^+(p)\cap M_k \neq \emptyset \; \; \mathrm{and} \;\; I^-(p)\cap M_k\neq \emptyset \]
		and
		\[A\setminus (I^+(p)\cup I^-(p)) \; \; \mathrm{is\; compact.} \]  
					
\item[($b_1$)] We require that there exist constants $c_{ij}>0$ and with $c_{00}<a$, such that for any $p\in M_1 \cup M_2$ and any $\go$-othonormal basis $\{e_i(p)\}_{i=0}^3\sse T_pM$, with $e_0 = \frac1{\cosh x}\frac{\partial}{\partial t}$, 
 \beq \label{eq:decaydef} 
 |h(e_i(p),e_j(p))|\leq \frac{c_{ij}}{|x(p)|},  
 \eeq
where $ h= g|_{M_1\cup M_2}-\go|_{M_1\cup M_2} $.

\item[($b_2$)] We further require the following decay on first derivatives of $h$, i.e., we assume  there exists $C_1 >0$ such that
\beq \label{eq:derivativeestimates}
 |e_k(h(e_i,e_j))(p) |\leq \frac{C_1}{|x(p)|}, \quad |e_0(h(e_i,e_j))(p)| \leq \frac{C_1}{|x(p)|^2}
\eeq 
for $k = 1,2,3$ (note the faster decay on the time derivative).  Additionally, we require the following decay on second derivatives,
\beq \label{eq:tderivativeestimates}
|e_l( e_m( h(e_i,e_j)))(p) |\leq \frac{C_1}{|x(p)|}
\eeq 
for $l,m=0,\dots ,3$.
	\een
\end{definition}

\smallskip

It will  be convenient to require, 
\beq\label{eq:Cestimate}\frac{16\max\{c_{ij}\}}{a}<1  \,.
\eeq

\begin{remark}\label{rem:decayincoords}
In many of the arguments involving the asymptotics \eqref{eq:decaydef}-\eqref{eq:tderivativeestimates} we will not use $\go$-othonormal frames but rather work in specific charts which we will now introduce. Let $(U,\psi)$ denote either of two charts covering $S^2$, with $\psi(U) = \{(\theta, \phi) : \frac{\pi}6 < \theta < \frac{5\pi}6$, $0 < \phi < 2\pi\}$,
and let $\psi: p \to (t(p), x(p), \psi(p))$   be the corresponding chart on $M_1\cup M_2$. From \eqref{eq:decaydef} we see that there exists a constant $C>0$ such that in these charts
\beq \label{eq:decaycoordinates} |h_{ij}|\leq \frac{C}{|x|},\; |h_{ti}|\leq\frac{C\cosh(x)}{|x|}, \; |h_{tt}|\leq \frac{C\cosh^2(x)}{|x|}  \eeq
for $i,j\neq t$. And for $l,m,i,j$ arbitrary and $k\neq t$
\beq \label{eq:decay'coordinates} |\partial_k h_{ij}|\leq \frac{C \cosh^{\#t}(x)}{|x|},\; |\partial_t h_{ij}|\leq \frac{C \cosh^{\#t}(x)}{|x|^2},\; |\partial_l \partial_m h_{ij}|\leq\frac{C\cosh^{\#t}(x)}{|x|} \eeq
where $\#t$ denotes the number of $t$'s appearing as lower indices.

We are also going to need some estimates for the Christoffel symbols and the curvature of $g$. Let now $\#t$ denote the number of $t$'s appearing as lower indices minus the number of $t$'s appearing as upper indices. Using \eqref{eq:decaycoordinates} and \eqref{eq:decay'coordinates} one can show that there exists a constant $C$ such that
\beq \label{eq:christofflasymptotics}
 |\go^{ij}-g^{ij}|\leq \frac{C\cosh^{\#t}(x)}{|x|},\; |\mathring{\Gamma}^k_{ij}-\Gamma^k_{ij}|\leq \frac{C\cosh^{\#t}(x)}{|x|}
\eeq
and
\beq \label{eq:curvatureasymptotics}
 \; |\mathring{R}_{iklm}-R_{iklm}|\leq \frac{C \cosh^{\#t}(x)}{|x|}, \;  |\mathring{\Ric}_{ij}-\Ric_{ij}|\leq \frac{C \cosh^{\#t}(x)}{|x|}, \; |\mathring{R}-R|\leq \frac{C}{|x|}.
\eeq
These estimates follow in a straightforward way from \eqref{eq:decaycoordinates} and \eqref{eq:decay'coordinates}, nevertheless their derivation is carried out in some detail in the appendix.
\end{remark}
\medskip
To simplify the constants appearing in later arguments we will always choose $C$ such that additionally 
\beq C\geq \max\{c_{ij}, C_1\}\,.\eeq

\subsection{Consequences of the asymptotics}

We will start by introducing some notations: First, for any $x_0\in \R$ with $|x_0|\geq a$ we will use the shorthand $\{x=x_0\}$ for the submanifold $\R \times \{x_0\} \times S^2$ of $M_k$ (where $k=1$ for $x_0<0$ and $k=2$ for $x_0>0$). Further, for $r\in \R_+$ and $k=1,2$ we use $M_k(r):=M_k \cap (\R \times \{x: |x|\geq r\} \times S^2)$ to denote the part of $M_k$ that lies between $\{|x|=r\}$ and infinity. We also set $M(r):=M_1(r)\cup M_2(r)$.

\begin{Lemma}\label{lem:betaalpha}
	For any $r\in[a,\infty)$ there exists $\beta_r >1$ and $\alpha_r<1$ such that on $M(r)$  

	\[\go_{\alpha_r}\prec g \prec \go_{\beta_r} \] and one can choose $\beta_r$ and $\alpha_r$ such that $\beta_r$ is decreasing in $r$ and $\alpha_r$ is increasing in $r$ and  $\beta_r, \alpha_r \to 1$ as $r\to \infty $.
\end{Lemma}
\begin{proof}
We first show the existence of a suitable $\alpha_r$. Let $\{e_i\}_{i=0}^3$ be a $\go $-orthonormal basis for $T_pM(r)$ and let $v=v^i e_i$ be such that $\go_\alpha (v,v)\leq 0$. We may w.l.o.g. assume  $\sum_{i=0}^3|v^i|^2=1$. 
Now, $\go_\alpha (v,v)\leq 0$ gives	\beq \alpha |v^0|^2 \geq \sum_{i=1}^{3} |v^i|^2= 1-|v^0|^2. \eeq Then
\begin{multline}
g(v,v)=\go (v,v)+ h(v,v)= -|v^0|^2 +\sum_{i=1}^{3}|v^i|^2+ h(v,v) \leq \\
\leq (\alpha -1) |v^0|^2 + h(v,v) \leq (\alpha -1) |v^0|^2+ \frac{16 C}{|x(p)|}\leq (\alpha -1) |v^0|^2+ \frac{16 C}{r}.
\end{multline}
Now if $\alpha <1$ we can use $|v^0|^2\geq \frac{1}{\alpha +1}$ to further estimate
\[g(v,v)\leq \frac{\alpha-1}{1+\alpha} + \frac{16C}{r}. \]
Thus, setting $\alpha_r<\frac{1-\frac{16C}{r}}{1+\frac{16C}{r}}<1$ guarantees $g(v,v)<0$ and since $\frac{1-\frac{16C}{r}}{1+\frac{16C}{r}} \to 0$ as $r\to \infty$ and is strictly decreasing we can choose $\alpha_r$ to be increasing and  $\alpha_r\to 1$.

For $\beta_r$ we note that it suffices to show that $\go_{\beta_r}(v,v)\geq 0$ implies $g(v,v)>0$. Now
$ \go_{\beta}(v,v)\leq 0$ gives $ 1-\sum_{i=1}^{3}|v^i|^2 =|v^0|^2 \leq \frac{1}{\beta} \sum_{i=1}^{3}|v^i|^2$. So we have
\begin{multline}
g(v,v)=\go(v,v)+h(v,v)\geq -|v^0|^2+\sum_{i=1}^{3}|v^i|^2 -|h(v,v)| \geq \\ \geq  (1-\frac{1}{\beta })\sum_{i=1}^{3}|v^i|^2 - \frac{16C}{r} \geq  \frac{1-\frac{1}{\beta }}{1+\frac{1}{\beta}} - \frac{16C}{r}.
\end{multline}
This implies the existence of a suitable $\beta_r$.
\end{proof}

This allows us to bound the time it takes for the entire $S^2$-factor to be contained in the future of a point depending on how far out (in the $x$-direction) this point lies.

\begin{Lemma}\label{lem:S2inIpaftertimetau}
	For any $r\in [a,\infty)$ there exists a time $\tau_r$ such that for any $p\in M_k(r)$
	\[\{t\geq t(p)+\tau_r\}\times \{x(p)\} \times S^2 \sse I^+(p),\]
	$\tau_r$ is decreasing in $r$ and $\tau_r \to 0$ as $r\to \infty $.
\end{Lemma}
\begin{proof}
	Let $\bar{\gamma}:I\to S^2$ be a unit speed geodesic (in $S^2$) starting at $\pi_{S^2}(p)$ and let $\alpha_r$ be the constant from the previous lemma. Since $|x(p)|\geq r$ the curve $\gamma(s):=(t(p)+\frac{1}{\sqrt{\alpha_r}\cosh(r)}s,x(p),\bar{\gamma}(s))$ is causal for $\go_{\alpha_r}$, hence timelike for $g$ by Lemma \ref{lem:betaalpha}.  Noting that $S^2$ has a finite diameter of $\pi$  proves the claim for $\tau_r:=\frac{\pi}{\sqrt{\alpha_r}\cosh(r)}$. This is decreasing and goes to zero as $r\to 0$ because $\alpha_r$ is increasing and $\alpha_r\to 1$.
\end{proof}

We also note the following consequence for null vectors.

\begin{Lemma} \label{lem:newlemma}
	Let $q_{ij}$, $i,j=0,\dots,3$, be smooth functions on $U\sse M_1\cup M_2$ satisfying the asymptotics $|q_{ij}|\leq \frac{C\cosh^{\#t}(x)}{|x|}$ (e.g., $q_{ij}=h_{ij},\mathring{\Gamma}^x_{ij}-\Gamma ^x_{ij},\mathring{\Ric}_{ij}-\Ric_{ij},\dots $). Then there exists a constant $c>0$ such that for any null vector $v\in TU$ we have
	\beq \label{eq:estfornullvec} |q_{ij} v^i v^j|\leq \frac{c}{|x|} (|v^x|^2+|\bar{v}|_{S^2}^2), \eeq 
	where $v^i$ denotes the components of $v$ in one of the charts $\psi $ specified in Remark \ref{rem:decayincoords}.
\end{Lemma}
\begin{proof}
		Let $\beta_a$ be as in Lemma \ref{lem:betaalpha}. Then $v$ being null implies $\go_{\alpha_a}(v,v)>0$, which gives the estimate $|v^t|^2<\frac{1}{\alpha_a \cosh^2(x)}(|v^x|^2+|\bar{v}|_{S^2}^2)$. Further, note that in either chart $\psi$ on $S^2$ one always has $|v^\theta| \le |\bar{v}|_{S^2}$ and 
		$|v^\phi|<2|\bar{v}|_{S^2}$, which gives the estimates,
		\begin{align} \label{eq: evilestimate}
		|v^i| |v^j| &\leq \frac{1}{2}(|v^i|^2+ |v^j|^2)\leq 4 (|v_x|^2+|\bar{v}|^2_{S^2}) \,, \\
			|v^t| |v^j| &\leq \frac{1}{\sqrt{\alpha_a} \cosh(x)}\sqrt{|v^x|^2+|\bar{v}|_{S^2}^2}  |v^j| \leq \frac{5}{2} \frac{1}{\sqrt{\alpha_a} \cosh(x)} (|v_x|^2+|\bar{v}|^2_{S^2}) 
		\end{align}
	for $i,j\neq t$. Hence \beq
	|q_{ij} v^i v^j|\leq \sum_{i,j\neq t} \frac{C}{|x|} |v^i| |v^j|+2 \sum_{j\neq t} \frac{C\cosh(x)}{|x|} |v^t| |v^j|+ \frac{C\cosh^2(x)}{|x|} |v^t|^2 \leq \frac{c}{|x|} (|v^x|^2+|\bar{v}|^2_{S^2}) \eeq
	for $c=(36+\frac{15}{\sqrt{\alpha_a}}+\frac{1}{\alpha_a})C$.
\end{proof}

Finally we want to study maximizing null curves. Generally we say that a null curve $\gamma :I\to M$ is a \emph{future (or past) null ray} if $I=[a,b)$ and $\gamma $ is maximizing (i.e., its image is achronal) and future (resp.\ past) directed and future (resp.\ past) inextendible. We say that a null curve $\gamma :I\to M$ is a \emph{null line} if $I=(a,b)$ and $\gamma $ is maximizing and inextendible in both directions.

\begin{Lemma}[Null rays must run to infinity]\label{lem:inextmaxnullrunstoinfty}
	Let $\gamma :I\to M$ be a future null ray.
Then $\gamma $ is eventually contained in one of the $M_k$'s and $|x(\gamma(s))|\to \infty $ as $s\to b$.
\end{Lemma}
\begin{proof}
	If $q\in \gamma \cap A$, then $\gamma $ must eventually leave the compact set $A\setminus (I^+(q)\cup I^-(q))$ and never return to it, but since $\gamma $ is achronal, that means that $\gamma$ cannot return to $A$ at all, i.e., it is contained in $M_k$, say $M_2$. So we may assume $\gamma(0)=(t_0,x_0,\omega_0)\in M_2$ and $x_0=a$. For any $r>0$ the set $[ t_0, t_0+\frac{\pi}{\sqrt{\alpha_a}}+\tau_a ] \times [a,r] \times S^2$
is compact, so $\gamma $ must leave it.  Since $\g$ is future directed, we must have $t > t_0$ along $\g$. Moreover, by applying, first, Lemma \ref{lem:t+piinIp}, then, Lemma \ref{lem:S2inIpaftertimetau}, together with Lemma   \ref{lem:betaalpha}, we see that we must have 
$t<t_0+\frac{\pi}{\sqrt{\alpha_a}}+\tau_a$ along $\g$, otherwise the achronality of $\g$ would be violated.
It follows that $\g$ must cross $x = r$.
\end{proof}

\begin{Lemma} \label{lem:noturningback} 
	There exists $r>0$ such that for any null geodesic $\gamma \sse M_1(r)$ one has $\ddot{\gamma}_x >0$, i.e., $\dot{\gamma}_x$ can change sign at most once.
\end{Lemma}
\begin{proof}
	By the geodesic equation and the estimate \eqref{eq:estfornullvec} we have
	\begin{align} \ddot{\gamma}^x &=-\Gamma^x_{ij}   \dot{\gamma}^i \dot{\gamma}^j\geq -\mathring{\Gamma}^x_{ij} \dot{\gamma}^i \dot{\gamma}^j- |(\mathring{\Gamma}^x_{ij}-\Gamma^x_{ij})  \dot{\gamma}^i \dot{\gamma}^j| \geq \\ 
	&\geq \cosh(|x|)\sinh(|x|) |v^t|^2 -\frac{c}{|x|} (|\dot{\gamma}_x|^2+|\dot{\bar{\gamma}}|^2_{S^2}) \end{align}
	Finally, $\gamma $ being null implies $\go_{\beta_r}(\dot{\gamma},\dot{\gamma})<0$, which gives $|\dot{\gamma}_t|^2>\frac{1}{\beta_r \cosh^2(x(\gamma))}(|\dot{\gamma}_x|^2+|\dot{\bar{\gamma}}|_{S^2}^2)$, so
	\beq \ddot{\gamma}^x> \big(\frac{1}{\beta_r} \tanh(|x|)-\frac{c}{|x|} \big) (|\dot{\gamma}_x|^2+|\dot{\bar{\gamma}}|^2_{S^2})>\big(\frac{1}{\beta_r} \tanh(r)-\frac{c}{r} \big) (|\dot{\gamma}_x|^2+|\dot{\bar{\gamma}}|^2_{S^2})> 0 \eeq for $r$ large.
\end{proof}

\begin{Corollary}\label{cor:xparam}  Let $r>0$ be such that the previous Lemma holds. Then any null geodesic $\gamma :[a,b)\to M_1(r)$ with $\lim_{s\to b}x(\gamma(s))=-\infty$ may be parametrized with respect to the $x$-coordinate. 
	\begin{proof}
		Let $\gamma :[a,b)\to M$ be a null geodesic
		 with image in $M_1(r)$. Lemma \ref{lem:noturningback} shows that $\dot{\gamma}^x$ is strictly increasing, so if $\dot{\gamma}^x(s_0)\geq 0$ for any $s_0\in [a,b)$ then $\dot{\gamma}^x|_{[s_0,b)}\geq 0$ and hence $x(\gamma|_{[s_0,b)})\geq x(\gamma(s_0))$. This contradicts $\lim_{s\to b}x(\gamma(s))= -\infty$. Thus $s\mapsto x(\gamma(s))$ is strictly monotonically decreasing and so there exists a reparametrization $\tilde{\gamma}:(-\infty,x(\gamma(a))]\to M_1(r)$ of $\gamma$ with $x(\tilde{\gamma}(s))=s$.
	\end{proof}

\end{Corollary}

\begin{Lemma} \label{lem:maxnulliscomplete}
	Any future (or past) null ray $\gamma :[0,a)\to M$ is future (or past) complete.
\end{Lemma}
\begin{proof}
	By the proof of Lemma \ref{lem:inextmaxnullrunstoinfty} for any $r>0$ $\gamma $ is eventually contained in either $M_1(r)$ or $M_2(r)$ and $|x(\gamma(s))|\to \infty $ as $s\to a$. For now, look at the case where $\gamma$ is eventually contained in $M_1(r)$ (for some large $r$, at least $r\geq r(1)$ from the previous Lemma). We may assume $\gamma :[0,a)\to M_1(r)$ and $|\dot{\gamma}^x(0)|=1$ and we have to show that $a=\infty $. By the arguments in Corollary \ref{cor:xparam} we have $\dot{\gamma}^x(0)<0$ and $s\mapsto \dot{\gamma}^x(s) $ is strictly increasing, so $|\dot{\gamma}^x(s)|^2 \leq 1$ for all $s$.
	But this gives $|x(\gamma(s))-x(\gamma(0))|\leq |s|$,  contradicting $x(\gamma(s))\to -\infty $ as $s\to a$ if $a< \infty$. The case of the end contained in $M_2(r)$ is analogous (note that the analogues to Lemma \ref{lem:noturningback} and Corollary \ref{cor:xparam} show $\ddot{\gamma}^x<0$ and $\dot{\gamma}^x>0$ on $M_2(r)$). 
	\end{proof}

\begin{Lemma}[Angular velocities go to zero for null lines] \label{lem:angularvelocitiesfornulllines}
	Assume the null energy condition holds, i.e., $\Ric(X,X) \ge 0$ for all null vectors $X$. For any $\eps >0$ there exists $r(\eps)$ such that $|\dot{\bar{\gamma}}|_{S^2} <\eps$ on $M_1(r(\eps))$  for any null line $\gamma :I\to M$ with $|\dot{\gamma}^x|\leq 1$ on $M_1(r(\eps))$.
\end{Lemma}
\begin{proof} 
	 Since $\gamma$ is complete by Lemma \ref{lem:maxnulliscomplete} the Raychaudhuri equation applied to $\g$ (affinely parametrized) implies that  $\Ric (\dot{\gamma},\dot{\gamma})=0$ along $\g$ (else $\gamma$ would contain a pair of conjugate points).  This condition does not change under reparametrization of $\gamma $.
We now use $\mathring{\Ric} (\dot{\gamma},\dot{\gamma})= -\go (\dot{\gamma},\dot{\gamma})+2 |\dot{\bar{\gamma}}|^2_{S^2}$ and \eqref{eq:estfornullvec} to estimate
	\begin{multline}
	0=\Ric (\dot{\gamma},\dot{\gamma}) \geq \mathring{\Ric} (\dot{\gamma},\dot{\gamma}) - \frac{c}{|x(\gamma(s))|} (|\dot{\gamma}^x|^2+|\dot{\bar{\gamma}}|^2_{S^2}) \geq  \\
	\geq -\go (\dot{\gamma},\dot{\gamma})+2 |\dot{\bar{\gamma}}|^2_{S^2}- \frac{c}{|x(\gamma(s))|} (1+|\dot{\bar{\gamma}}|^2_{S^2}) 
	\geq -g(\dot{\gamma},\dot{\gamma})+2 |\dot{\bar{\gamma}}|^2_{S^2}- \frac{2c}{|x(\gamma(s))|} (1+|\dot{\bar{\gamma}}|^2_{S^2})= \\ =2 |\dot{\bar{\gamma}}|^2_{S^2}- \frac{2c}{|x(\gamma(s))|} (1+|\dot{\bar{\gamma}}|^2_{S^2}).
	\end{multline}
	So 
	\beq \big(2-\frac{c}{|x(\gamma(s))|}\big)|\dot{\bar{\gamma}}|^2_{S^2}\leq \frac{c}{|x(\gamma(s))|}
	\eeq
from which the claim follows.
\end{proof}

\section{Proof of the main results}
\label{sec:main}

	Throughout this section we will frequently make use of the null energy condition,
	$\Ric(X,X) \ge 0$ for all null vectors $X$. This assumption enters in Proposition \ref{prop:sametlimit} (and thus Remark \ref{rem:sametlimit2lines}) via Lemma \ref{lem:angularvelocitiesfornulllines} and in Theorem \ref{thm:existenceofNullHyp} via both Remark \ref{rem:sametlimit2lines} and \cite[Theorem~IV.1]{G00}. All further results, in particular all of subsection \ref{sec:spherefoliation}, build upon Theorem \ref{thm:existenceofNullHyp}.

\subsection{Constructing a foliation by totally geodesic null hypersurfaces}

\begin{Lemma}\label{lem:Ipmeetseveryxslice}
	Let $p\in M_k$ and $x_0\in (-\infty, -a]\cup[a,\infty)$. Then $I^\pm (p)\cap \{x=x_0\}\neq \emptyset $. 
\end{Lemma}
\begin{proof}
	Let w.l.o.g. $p\in M_1$ and first consider  $x_0\in (-\infty, -a]$. Then this is clearly true for any $\go_{\alpha}$, hence by Lemma \ref{lem:betaalpha} also for $g$. Since $\{x=-a\}\sse \partial A \sse A$, condition~($a_2$) from Def. \ref{def:assumptions} then shows that $I^\pm (p)\cap M_2 \neq \emptyset $. Note that this also must even imply $I^\pm (p)\cap \{x=a\} \neq \emptyset $ from which the claim follows for $x_0 \in [a,\infty)$ by the same argument as above.
\end{proof}

\begin{Lemma}\label{lem:forpinM1existsmaxnull}
	For any $p\in M_1$ there exists a future null ray $\gamma_p :[0,b)\to M$ such that $\gamma_p $ is eventually contained in $M_2$.
\end{Lemma}
\begin{proof}
	By Lemma \ref{lem:Ipmeetseveryxslice} and causality, for each positive integer $n \in [a, \infty)$,
$\{x=n\}\neq I^+(p)\cap \{x=n\}\neq \emptyset$, so there exist $q_n\in \partial J^+(p)\cap \{x=n\}$. Every $q_n$ is the future endpoint of a maximizing null geodesic $\gamma_n \sse \partial J^+(p)$ which must end in $p$ because $J^+(p)$ is closed. Hence, there exists a limit curve $\gamma $ starting at $p$ that is maximizing and inextendible (because the $q_n$ run off to infinity). It is eventually contained in $M_2$ because by the `no turning back' lemma (Lemma \ref{lem:noturningback}) $\gamma_n\cap \{x=-r\}=\emptyset$ for $r$ large, so $|x(\gamma(s))|\to \infty$ as $s\to \infty$ (see Lemma \ref{lem:inextmaxnullrunstoinfty})
implies $x(\gamma(s))\to \infty$. 
\end{proof}

This allows us to construct null lines:

\begin{Proposition} \label{prop:existenceofnulllines}
	For any $u \in \R$ there exists a complete null line $\eta_u:(-\infty,\infty)\to M$ with past end eventually contained in $M_1$ and future end eventually contained in $M_2$ and $t(\eta_{u}(s))\to u$ as $s\to -\infty$.
\end{Proposition}
\begin{proof}
	Let $u\in \R$, fix any $\omega_0 \in S^2$ and set 
	$p_n:=(u,-n,\omega_0)\in M_1$. Then by Lemma \ref{lem:forpinM1existsmaxnull} there exist maximizing future inextendible null curves $\gamma_n:[0,\infty)\to M$ 
	starting at $p_n$ that are eventually contained in $M_2$. We now show that the sequence $\gamma_n$ contains an accumulation point. Let $t_{n,m}$ be the maximal $t$-coordinate of the set $\gamma_n \cap \{x=-m\}\neq \emptyset $ for $n\geq m$. Clearly $t_{n,m}\geq u$ for $n\geq m$. By Lemmas \ref{lem:betaalpha}, \ref{lem:nulllinseforgbeta} (and remark \ref{rem:behaviouroff}) and \ref{lem:S2inIpaftertimetau} we see that $ t_{n,m}\leq u+\tau_m+\frac{\pi}{\sqrt{\alpha_m}}-\frac{2\tan^{-1}(e^{-n})}{\sqrt{\alpha_m}} <u+\tau_m+\frac{\pi}{\sqrt{\alpha_m}}<c$
(because all points $p$ in $\{x=-m\}$ with larger $t$-coordinate belong to $I^+(p_n)$). Thus the sequence $\{t_{n,m}\}_\red{{n\geq m}}$ has an accumulation point for $m$ large.
	
	By the no turning back lemma, for large enough $m$ each $\gamma_n$ meets $\{x=-m\}$ in a unique point. We reparametrize such that this point is always $\gamma_n(0)$. 
	
	Thus there exists a limit curve $\eta_u$ which is maximizing and both past and future inextendible, hence complete by Lemma \ref{lem:maxnulliscomplete}. Since $\gamma_n|_{[0,\infty)}\sse \{x\geq -m\}\cup A \cup M_2$ the same holds for $\gamma |_{[0,\infty)}$, so the future end of $\gamma $ is eventually contained in $M_2$ (by a similar argument to Lemma \ref{lem:forpinM1existsmaxnull}). And since $\gamma_n|_{[a_n,0]}\sse \{x\leq -m\} \sse M_1$ the past end of $\gamma $ must lie in $M_1$.
	
	Finally, we need to argue that $t(\eta_{u}(s))\to u$ as $s\to -\infty$. Since $t(\gamma_n)\geq u$ (as long as $\gamma_n$ remains in $M_1$) the same holds for $t(\eta_u)$. Assume now that $t(\eta_{u}(s))\to u_1$ with $u_1>u$. This implies that $t(\eta_{u}(s))>u_1>u+\eps$ for all $s$. But this is a contradiction to $t_{n,m}\leq f_{\alpha_m}(n-m,u,-n)+\tau_m = \tau_m +\frac{\pi}{\sqrt{\alpha_m}}(\tan^{-1}(e^{-m})-\tan^{-1}(e^{-n}))+u< u+\eps$ for $m$ large and $n\geq m$. 
	\end{proof}

While the construction above depends on the choice of $\omega_0\in S^2$ and hence is not unique, we are now going to argue that any null line $\eta $ with $t(\eta(s))\to u$ is contained in a totally geodesic null hypersurface $N_u$ that \emph{only} depends on $u$. We first note the following:

\begin{Proposition} \label{prop:differenttimelimitmeansIp}
	Given two past inextendible causal curves $\eta_1,\eta_2 :(-\infty,0]\to M$ with past end contained in $M_1$, $\lim_{s\to -\infty}x(\eta_i(s))= -\infty $ and $\lim_{s\to-\infty} t(\eta_1(s))>\lim_{s\to-\infty} t(\eta_2(s))$ one has $\eta_1 \sse I^+(\eta_2)$.
\end{Proposition}
\begin{proof}
	Use Lemma \ref{lem:S2inIpaftertimetau}, note that $\tau_r \to 0$ as $r\to \infty$ and that by assumption $|x(\eta_i(s))|\to \infty $ as $s\to -\infty$.
\end{proof}

\begin{Proposition} \label{prop:sametlimit}
	Let $\eta_1: (-\infty, \infty )\to M$ be a null line and $\eta_2:(-\infty,b]\to M$ be a past null ray, both with past end contained in $M_1$, such that $\lim_{s\to-\infty} t(\eta_1(s))=\lim_{s\to-\infty} t(\eta_2(s))$. Then $\eta_1 \sse \overline{I^+(\eta_2)}$. If $\eta_2$ extends to a null line this further implies $\eta_1 \sse \partial J^+(\eta_2)$ and vice versa.
\end{Proposition}

\begin{proof}
	We may assume that, far enough out, both curves are parametrized with respect to the $x$-coordinate, so $|\dot{\eta}^x_{1,2}|=1$. We also note that by Lemma \ref{lem:angularvelocitiesfornulllines} since $\eta_1$  is assumed to be maximizing and both future and past  inextendibile we have that $|\dot{\bar{\eta}}_{1}(s)|_{S^2}^2 \leq 1$.
	
	We will first show that $\eta_1 \sse \overline{I^+(\eta_2)}$. This follows immediately if we can find $r>0$ such that $\eta_1 |_{(-\infty,-r]}$ can be approximated by curves $^\delta\eta_1 \sse I^+(\eta_2)$. We are now going to construct such approximating curves.
	
	To do this we estimate $|g_{(t_1,x,\theta,\phi)}-g_{(t_2,x,\theta,\phi)}|$ in terms of $|t_1-t_2|$. Since $\go $ is independent of $t$ this is just the difference of the corresponding $h$-terms and and $\frac{1}{\cosh(x)}|\partial_th_{ij}|\leq \frac{C\cosh^{\#t(i,j)}}{|x|^2}$, \eqref{eq:estfornullvec} gives (assuming $v$ is null and satisfies the same estimates as $\dot{\eta}_1$)
	\beq
	|h_{(t_1,x,\theta,\phi)}(v,v)-h_{(t_2,x,\theta,\phi)}(v,v)|\leq |\partial_t h_{ij}| |v^i| |v^j| |t_1-t_2|\leq \frac{2c\cosh(x)}{x^2} |t_1-t_2|\,.
	\eeq
	
	For a function $f>0,\dot{f}>0$ (which will be determined later), we define the curve $^\delta \eta_1 :=(\eta_{1}^t(s)+\delta+\delta f(s), s, \bar{\eta}_1(s))$. Clearly these curves approximate $\eta_{1}$. And by the above and $\eta_1$ being null we may estimate $g_{^\delta \eta_1(s)}(\dot{\eta}_{1}(s),\dot{\eta_{1}}(s))\leq \frac{2c\cosh(s)}{s^2} (\delta + \delta f(s))$. This, together with $\frac{1}{\beta_s \cosh^2(s)}\leq |\dot{\eta}^t_{1}(s)|^2\leq \frac{2}{\alpha_s \cosh^2(s)}$ , leads to
	\begin{multline}
g_{^\delta\eta_1}(^\delta \dot{\eta}_{1},^\delta \dot{\eta_{1}})\leq \frac{2c\cosh(s)}{s^2} (\delta + \delta f(s))+2 \delta g_{^\delta\eta_1}(\dot{\eta}_{1}, (\dot{f},0,\bar{0})) +\delta^2\dot{f}^2 g_{tt} \leq \\ \leq  \frac{2c\cosh(s)}{s^2} (\delta + \delta f(s))+2 \delta \go(\dot{\eta}_{1}, (\dot{f},0,\bar{0})) + |2\delta \dot{f} h_{it}  \dot{\eta}_1^i| \leq \\ \leq \frac{2c\cosh(s)}{s^2} (\delta + \delta f(s))-  \frac{2 \delta \cosh^2(s)}{\sqrt{\beta_s} \cosh(s) } \dot{f}(s)+ |2\delta \dot{f} h_{it}  \dot{\eta}_1^i|\,.
	\end{multline}
 Finally, $|h_{it} \dot{\eta}_1^i|\leq |h_{tt}| |\dot{\eta}_1^t|+\sum_{i\neq t} |h_{it}|\leq (\frac{\sqrt{2}}{\sqrt{\alpha_s}}+3)\frac{C\cosh(s)}{|s|}\leq \frac{5C\cosh(s)}{|s|}$ for $s$ large enough to ensure $\alpha_s>\frac{1}{2}$. So
 \beq g_{^\delta\eta_1}(^\delta \dot{\eta}_{1},^\delta \dot{\eta_{1}})\leq \delta \cosh(s) \big(\frac{2c}{s^2}(1+f(s))-\frac{2}{\sqrt{\beta_s}}\dot{f}(s)+\frac{5C}{|s|}\dot{f}(s)\big)\eeq 
	
	Now if $f(s)=|s|^{-\kappa}$ with $0<\kappa <1$ we have that $f$ is bounded by one and $\dot{f}(s)\to 0$ slower than $\frac{1}{s^2}$ as $s\to -\infty$. So there exists $r$ (independent of $\delta$) 
	 such that $^\delta \eta_1|_{(-\infty,-r]}$ is timelike. Since by construction $\lim t(^\delta \eta_1(s))=\lim t(\eta_1(s))+\delta=\lim t(\eta_2(s))+\delta$ Proposition \ref{prop:differenttimelimitmeansIp} implies $^\delta\eta_1 \sse I^+(\eta_2)$.
	
	If both curves are null lines, we may apply the same argument to $\eta_2 $ to get get $\eta_1  \sse \overline{I^+(\eta_2)}$ and $\eta_2  \sse \overline{I^+(\eta_1)}$. From this we see that $I^+(\eta_2)\sse I^+(\overline{I^+(\eta_1)})=I^+(\eta_1)$, hence $\eta_1 \cap I^+(\eta_2)=\emptyset $ by achronality of $\eta_1$, proving the claim.
\end{proof}

\begin{remark} \label{rem:sametlimit2lines}
	Note that this implies that for any future directed null line $\eta $  with past end contained in $M_1$ the set $\partial J^+(\eta)$ depends only on $\lim_{s\to -\infty} t( \eta (s))$. In particular for two lines $\eta_{u,\omega_1}$ and $\eta_{u,\omega_2}$ constructed as in Proposition \ref{prop:existenceofnulllines} one has $\partial J^+(\eta_{u,\omega_1})=\partial J^+(\eta_{u,\omega_1})=:\partial J^+(\eta_{u})$.
\end{remark}

\begin{Proposition}\label{prop:homeomS2}
	For any $u\in \R $ and any $x_0\in (-\infty,-a]\cup [a,\infty)$ the set $\partial J^+(\eta_{u}) \cap \{x=x_0\} \sse \R \times \{x_0\} \times S^2$ is a  graph over $S^2$ with continuous graphing function 
$\mathcal{T}_{u, x_0}: S^2\to \R$.  In particular, it is connected.
\end{Proposition}

\begin{proof}
	Let $\pi : \{x=x_0\} \equiv \R \times \{x_0\} \times S^2 \to S^2$ be the projection onto $S^2$ and define $S:=\partial J^+(\eta_{u}) \cap \{x=x_0\}$.   Being the intersection of an achronal  locally Lipschitz hypersurface with a timelike hypersurface, $S$ is itself an achronal locally Lipschitz hypersurface in $\{ x = x_0\}$. 
	Clearly, $\pi|_S $ is injective since $S$ is achronal and $\partial_t$ is timelike. Hence we may define $\mathcal{F}_{u,x_0}:=(\pi|_S )^{-1}: \pi(S) \sse S^2 \to S$.
	Next we will argue that $S$ is actually compact: Let $(t_0,x_0,\omega_0)\in S$. Then by Lemma \ref{lem:S2inIpaftertimetau}, any $p\in \{x=x_0\}$ with $t(p)<t_0-\tau_{|x_0|}$  lies in $I^-(S)$ and any  $p\in \{x=x_0\}$ with $t(p)>t_0+\tau_{|x_0|}$  lies in $I^+(S)$. So by achronality $S\sse [t_0-\tau_{|x_0|}, t_0+\tau_{|x_0|}] \times \{x_0\} \times S^2$, hence it must be compact. This implies that $\pi_S : S\to \pi(S)\sse S^2$ is actually a homeomorphism onto its image. In particular $\pi(S)$ is compact and $\mathcal{F}_{u,x_0}$ is continuous. Since $S$ is itself a two dimensional (topological) manifold, invariance of domain implies that $\pi|_S :S\to S^2$ is an open map. Hence  $\pi(S)= S^2$.
	
Thus, $\mathcal{F}_{u,x_0}$ is a homeomorphism,  and hence $S$ is homeomorphic to $S^2$, in particular connected. The graphing function  $\mathcal{T}_{u, x_0}: S^2 \to \R$, defined via: 
$\mathcal{T}_{u, x_0}(\omega) = t(\mathcal{F}_{u,x_0}(\omega))$ is clearly continuous.	
\end{proof}

\begin{Corollary}\label{cor:connected}
	For any $u\in \R $ the set $\partial J^+(\eta_{u})$ has only one connected component.
\end{Corollary}
\begin{proof}
	Any point in $\partial J^+(\eta_{u})$ lies on a past inextendible achronal null geodesic $\gamma_p$ contained in $\partial J^+(\eta_{u})$. By the time dual of Lemma \ref{lem:inextmaxnullrunstoinfty} we know that $\gamma_p$ eventually enters $M_1$ or $M_2$ and hence meets $\{x=x_0\}$ for some $x_0\in (-\infty,-a]\cup[a,\infty )$. Now since $\eta_{u}$ meets every $\{x=\mathrm{const.}\}$ slice and $\{x=x_0\}\cap \partial J^+(\eta_{u})$ is connected, $p$ lies in the same connected component as $\eta_{u} $. Since this is true for every $p$, connectedness follows.
\end{proof}

\begin{Theorem}\label{thm:existenceofNullHyp}
	For any $u\in \R$ there exists a smooth closed achronal totally geodesic null hypersurface $N_{u}$ such that there exists a null geodesic generator $\eta$ with $u=\lim_{s\to-\infty} t(\eta(s))$. Further $\lim_{s\to-\infty} t(\eta(s))$ is independent of the choice of the null generator $\eta $ and determines $N_{u}$ uniquely. We have $N_u=\partial J^+(\eta )=\partial J^-(\eta )$.
\end{Theorem}

\begin{proof}
	Let $\eta_{u}$ be any of the null lines from Prop.~\ref{prop:existenceofnulllines}. Note that by Lemma \ref{lem:maxnulliscomplete} the null geodesic generators of $\partial J^+(\eta_{u})$ and $\partial J^-(\eta_{u})$ are complete, so we may apply the null splitting theorem \cite[Theorem~IV.1]{G00} to $\eta_{u}$ by \cite[Remark~IV.2]{G00}. This gives that the connected component of $\partial J^+(\eta_{u})$ containing $\eta_{u}$ is a smooth closed achronal totally geodesic null hypersurface which by construction contains a null geodesic generator $\eta$ with $u=\lim_{s\to-\infty} t(\eta(s))$. Now since $\partial J^+(\eta_{u})$ and $\partial J^-(\eta_u)$ are connected (the same arguments as in Prop.~\ref{prop:homeomS2} and Cor.~\ref{cor:connected} also give connectedness of $\partial J^-(\eta_u)$) the null splitting theorem further shows $\partial J^+(\eta_u)=\partial J^-(\eta_u)$. The remaining claims follow from Prop.~\ref{prop:differenttimelimitmeansIp} and Remark \ref{rem:sametlimit2lines}. Hence, $N_{u}:=\partial J^+(\eta_{u})$ is the null hypersurface we were looking for.
\end{proof}

\begin{Theorem}\label{thm:noholes}
For any $p\in M$ there exists a \emph{unique} $u_p\in \R$ such that $p\in N_{u_p}$.
\end{Theorem}

\begin{proof}
	That $N_{u_1}\cap N_{u_2}=\emptyset $ for $u_1 \neq u_2$ is clear from Theorem~\ref{thm:existenceofNullHyp}. We need to show that for any $p\in M$ there exists $u_p$ such that $p\in N_{u_p}$. We start by showing that there exists $r>0$ such that $M(r)$ is covered by the union $\bigcup_{u\in \R} N_u$.
	
	Let $\sigma =\R \times \{x_0\} \times \{\omega_0\} \sse M(r)$ be any $t$-line in $M(r)$. We define a function 
$f : \R \to \R$ as follows: For each $u \in \R$ there is an associated totally geodesic null hypersurface $N_u$.  By Proposition \ref{prop:homeomS2}, $N_u$ meets $\sigma$ in a unique point; let $f(u)$ be the $t$-coordinate of that point.  Using Proposition \ref{prop:differenttimelimitmeansIp}, one sees (in order to avoid an achronality violation) that $f$ is strictly increasing.  

We will now argue that $f$ is continuous and onto. Fix an interval $[a,b]$, and let $u_0 \in (a,b)$. We have $f(a) < f(u_0) < f(b)$.  We first show that $f$ is continuous from the left, i.e.,
$\lim_{u \to u_0^-} f(u) = f(u_0)$.

To each $u < u_0$ we have an associated null hypersurface $N_u$, and hence an associated null geodesic generator $\eta_u$ determined by where $N_u$ meets $\s$.  Note, for $u_1 < u_2$, we have 
$\eta_{u_2} \subset I^+(\eta_{u_1})$.  By considering their intersection with $\s$ and noting that $f(a)<f(u)=t(\eta_u\cap \s)<f(b)$, we see that as $u \nearrow u_0$ the null lines $\eta_u$  accumulate to a unique null line $\eta$ passing through $\sigma$ at a $t$-coordinate $t=\sup\{f(u): u<u_0\}=\lim_{u\to u_0^-} f(u)$.  By the null splitting thorem, $\eta$ determines a totally geodesic null hypersurface $N_v$ for some $v$.  Then $\eta = \eta_v$ is the null geodesic generator of $N_v$ determined by where $N_v$ meets $\s$. Clearly we must have $v \ge u_0$, otherwise, by Proposition  \ref{prop:differenttimelimitmeansIp}, $\eta_u$ would lie to the future of $\eta_v$ for $u$ sufficiently close to $u_0$,  which would contradict $f(u)<f(v)$.   If it were the case that $v > u_0$ then $\eta_v$ would be in the timelike future of 
$\eta_{u_0}$.  But then, $f(v)>f(u_0)$, so by the convergence, $f(u)>f(u_0)$  for $u$ sufficiently close to $u_0$, contradicting monotonicity of $f$.  Hence, 
$v = u_0$, and we conclude that $\lim_{u \to u_0^-} f(u) = f(u_0)$.  

A similar argument shows $\lim_{u \to u_0^+} f(u) = f(u_0)$.  Thus for any $a < c < d < b$, $f$ is continuous on $[c, d]$, and, since increasing, onto $[f(c), f(d)]$.  Since $[a,b]$ is arbitrary, this is enough to imply the claim.

	Thus, we have shown that every $t$-line in $M(r)$ is covered by $\cup_{u \in \R} N_u$, so $M(r)\sse \cup_{u \in \R} N_u$. Let now $p\in M$ be arbitrary. By a dual argument to Lemma \ref{lem:forpinM1existsmaxnull} there exists a past inextendible maximizing ray $\gamma_p$ that is eventually contained in $M_1(r)\sse \cup_{u \in \R} N_u$. Now for any $s_0$ with $\gamma_p|_{(-\infty,s_0]}\sse M_1(r)$ either $\gamma_p|_{(-\infty,s_0]}\sse N_{u_p}$ for $u_p:=\lim_{s\to -\infty}t(\gamma_p(s))$, then $p\in N_{u_p}$ since $N_u$ is totally geodesic and we are done. Or there exists $u \neq u_p$ with $\gamma_p|_{(-\infty,s_0]}\cap N_{u} \neq \emptyset $.
	If $u>u_p$ this contradicts achronality of $\gamma_p$ because by Prop.~\ref{prop:differenttimelimitmeansIp} $N_{u}\sse I^+(\gamma_p)$. If $u <u_p$ this contradicts achronality of $N_{u}$ because $\gamma_p \sse I^+(N_{u})$.
\end{proof}

\begin{remark}\label{rem:disjoint union}
From this we get the following structure: For any $u\in \R$ the spacetime $M$ is the disjoint union of $I^+(N_u)$, $N_u$ and $I^-(N_u)$. Let $p\in M$, then $p\in N_{u_p}$ for some $u_p$. If $u_p=u$, then $p\in N_u$. If $u_p>u$, then $p\in I^+(N_u)$ by Prop.~\ref{prop:differenttimelimitmeansIp}. Finally, if $u_p<u$, then $I^-(N_u)\cap N_{u_p}\neq \emptyset$ (by the argument in the proof of Prop.~\ref{prop:differenttimelimitmeansIp}). Thus $N_{u_p}\sse I^-(N_{u_p})$ since $N_{u_p}\cap \partial I^-(N_u)=\emptyset $ (because $\partial I^-(N_u)=N_u$ by Theorem~\ref{thm:existenceofNullHyp} and $N_u\cap N_{u_p}=\emptyset$).
\end{remark}

\begin{Theorem}
	The null hypersurfaces $\{N_u:u\in \R \}$ form a continuous codimension one foliation of $M$.
\end{Theorem}

\begin{proof}
	Let $(t, x_1,x_2,x_3)$ be coordinates on some open set  $U$ with $\partial_t$ timelike. We will show that $\psi:U\to \R^{4}$ defined by $\psi (p):=(u_p, x_1(p),x_2(p) ,x_3(p))$ is a continuous chart on $U$, for which clearly $\{p\in U: u_p=c\}=N_c \cap U$.
	Further, $p\mapsto u_p$ is continuous on $M$: Let $p_n\to p_0$, then the null lines $\eta_n\sse N_{u_{p_n}}$ corresponding to $p_n$ accumulate to a null line $\eta\sse N_{u_{p_0}}$ at $p_0$. From this continuity follows as in the previous proof.
Finally, $\psi $ is injective. Assume $\psi(p_1)=\psi (p_2)$, then $x_i(p_1)=x_i(p_2)$ for $i=1,2,3$ and it remains to show that $t(p_1)=t(p_2)$. If not, w.l.o.g. $t(p_1)>t(p_2)$ so, by $t$ being the time coordinate,  $p_1\in I^+(p_2)$ which contradicts $u_{p_1}=u_{p_2}$ by achronality of the $N_u$'s.
	From this invariance of domain implies that $\psi $ is a homeomorphism, i.e., a continuous chart.
\end{proof}

\subsection{Obtaining a foliation by totally geodesic round $2$-spheres}\label{sec:spherefoliation}

The same way one constructed the foliation $\{N_u\}_{u\in \R}$ one may obtain a second, transverse foliation with the same properties except that its null geodesic generators will be past instead of future directed. We denote this transverse foliation by $\{\hat{N}_v\}_{v\in \R}$. The idea is now to show that $S_{u,v}:=N_u\cap \hat{N}_v$ (if non-empty) are isometric 2-spheres and to use the asymptotics to argue that they must even be isometric to round 2-spheres.

We will first aim to  characterize the pairs $(u,v)$ for which $S_{u,v}\neq \emptyset$. To do so, let $\eta_u $ be a future directed null geodesic generator of $N_u$ and $\hat{\eta}_v$ be a past directed null geodesic generator for $\hat{N}_v$. Then we define $u_\infty :=\lim_{s\to \infty}t(\eta_u(s))$ and $v_\infty :=\lim_{s\to \infty}t(\hat{\eta}_v)$. These do not depend on the choice of $\eta_u,\hat{\eta}_v$ by an analogue of Prop. \ref{prop:sametlimit}.

\begin{Lemma} \label{lem:Suvnonempty} Let $(u,v)\in \R^2$. Then the following are equivalent: 
	\ben \item $S_{u,v}\neq \emptyset $, 
	\item $u<v$ and $u_\infty >v_\infty$,
	\item any (future directed) null geodesic generator of $N_u$ starts in $I^-(\hat{N}_v)$ and ends in $I^+(\hat{N}_v)$,
	\item any null geodesic generator of $N_u$ meets $S_{u,v}$ exactly once.
	\een
\end{Lemma}
\begin{proof}
	We begin by showing that $(i)$ implies $(ii)$. If $S_{u,v}\neq \emptyset$ then there exists a (future directed) null geodesic generator  $\eta_u$ of $N_u$ with $\eta_u\cap \hat{N}_v\neq \emptyset $. Since $\eta_u$ intersects $\hat{N}_v$ transversally and $M=I^+(\hat{N}_v)\cup \hat{N}_v \cup I^-(\hat{N}_v)$ (by the analogue of Remark \ref{rem:disjoint union}) it intersects $\hat{N}_v$ only once, say in $\eta_u(s_0)$. Because it is future directed we have $\eta_u|_{(-\infty,s_0)}\sse I^-(\hat{N}_v)$ and $\eta_u|_{(s_0,\infty)}\sse I^+(\hat{N}_v)$. So there exists $r$ large such that $\eta_u\cap M_1(r) \sse I^-(\hat{N}_v)$. Then $t(\eta_u \cap \{x=-r\})<t(\hat{\eta}_v\cap \{x=-r\})$ for an appropriately chosen (past directed) generator of $\hat{N}_v$. This gives $u<v$ because $t$ is decreasing along $\eta_u$ and increasing along $\hat{\eta}_v$ as $s\to -\infty $. An analogous argument in $M_2(r)$ shows $u_\infty >v_\infty $.
	
	Now, if $u<v$ and $u_\infty >v_\infty $ it immediately follows from (a slight variation of) Prop.~\ref{prop:differenttimelimitmeansIp} that any null geodesic generator of $N_u$ starts in $I^-(\hat{N}_v)$ and ends in $I^+(\hat{N}_v)$. This shows $(iii)$.
	
	If any null geodesic generator of $N_u$ starts in $I^-(\hat{N}_v)$ and ends in $I^+(\hat{N}_v)$ then it must intersect $\hat{N}_v=\partial I^+(\hat{N}_v)$ and hence $S_{u,v}$ at least once. Further, it can intersect $S_{u,v}$ at most once by the same argument as in the first paragraph. This shows $(iv)$.
	
	Finally, that $(iv)$ implies $(i)$ is obvious.
\end{proof}

\begin{Proposition}\label{prop:isometric}
	For any $(u,v)\in \R^2$ with $u<v$ and $u_\infty>v_\infty$ the set $S_{u,v}$ is a totally geodesic, spacelike codimension $2$ submanifold homeomorphic to $S^2$. Further for any two such pairs $u_1,v_1$ and $u_2,v_2$ the spheres $S_{u_1,v_1}$ and $S_{u_2,v_2}$ are isometric.
\end{Proposition}

\begin{proof}
	That the intersection is a totally geodesic, (smooth) spacelike codimension $2$ submanifold follwos immediately from $N_u$ and $\hat{N}_v$ intersecting transversally and being totally geodesic.
	
	Let $n_u$ be a null vectorfield defining $N_u$ with $\nabla_{n_u}n_u=0$. Then its flow $\Phi^{n_u}:\R \times (N_u\cap \{x=-r\}) \to N_u$ is a diffeomorphism (for $r$ sufficiently large): By Lemma \ref{lem:inextmaxnullrunstoinfty} and Corollary \ref{cor:xparam} every integral curve of $n_u$ intersects $N_u\cap \{x=-r\}$ exactly once and clearly every point of $N_u$ lies on an integral curve. Since by Lemma \ref{lem:Suvnonempty} any integral curve also intesects $S_{u,v}$ exactly once we may rescale $n_u$ such that 
	$\Phi^{n_u}(1,.):\{x=-r\}\cap N_u \to S_{u,v}$ is a diffeomorphism. Thus $S_{u,v}$ is homeomorphic to $S^2$ by Prop. \ref{prop:homeomS2}.
	
	Next we show that $S_{u,v_1}$ is isometric to $S_{u,v_2}$ if both are non-empty. This follows by a fairly standard argument from the fact that $N_u$ is totally geodesic  (see e.g \cite[Appendix A]{CDGH}):
We rescale $n_{u}$ such that $\Phi_1^{n_u}\equiv \Phi^{n_{u}}(1,.):S_{u,v_1}\to S_{u,v_2}$ is a diffeomorphism. Let $X_1,X_2$ be a basis for $T_pS_{u,v_1}$. We need to show that $g(X_i,X_j)=g((\Phi^{n_u}_1)_*(X_i),(\Phi^{n_u}_1)_*(X_j))$. For $s\in (0,1]$ we set $X_i(s):=(\Phi^{n_u}_s)_*X_i$. Then a straightforward computation shows $\mathcal{L}_{n_u}X_i=0$. Setting $g_{ij}(s)=g(X_i(s),X_j(s))$ we have
	\begin{align} \frac{d}{ds}g_{ij}&=n_u(g(X_i,X_j))=g(\nabla_{n_u}X_i,X_j)+g(X_i,\nabla_{n_u}X_j)=\nonumber \\&=g(\nabla_{X_i}n_u,X_j)+g(X_i,\nabla_{X_j}n_u)=0, \end{align}
	because the null second fundamental form of $N_u$ vanishes since $N_u$ is totally geodesic. The same argument (only  using $\hat{N}_{v}$ instead of $N_u$) applies to 
	show that $S_{u_1,v}$ and $S_{u_2,v}$ are also isometric. Since one can see that any two (non-empty) spheres $S_{u_1,v_1}$ and $S_{u_2,v_2}$ can be 
	connected via finitely many steps of this form. 
\end{proof}

Now we will estimate the curvature of such spheres $S_{u,v}\sse M(r)$ for $r$ large.

\begin{Proposition}\label{prop:gaussCurvature}
	For any $\eps >0$ there exists $r$ (depending only on $\eps$) such that any (non-empty) $S_{u,v}\sse M(r)$ has Gauss curvature $1-\eps \leq K_{u,v} \leq 1+\eps$.
\end{Proposition}
\begin{proof}
	Let $n_u$ and $\hat{n}_v$ be null vector fields defining $N_u$ and $\hat{N}_v$, respectively. We assume that they are normalized to $n^x_u=\hat{n}^x_v=1$. To simplify notation, we will drop the indices $u,v$. Let $X\in TS$ be any vector tangent to $S$ with $|X^x|^2+|\bar{X}|_{S^2}^2\leq 1$ and hence in our charts $|X^\theta|\leq 1,|X^\phi|\leq 2$. We will first estimate $|X^t|$: Since $X\in TS$ we have $g(X,n-\hat{n})=0$ and $X$ is $g$-spacelike, so $\go_{\alpha_r}(X,X)>0$, i.e., $|X^t|\leq \frac{1}{\sqrt{\alpha_r}\cosh(x)}\leq \frac{2}{ \cosh(x)}$ for $r$ large. Thus, we estimate
	\begin{multline}
	|\go(X,n-\hat{n})-g_{S^2}(\bar{X},\bar{n}-\bar{\hat{n}})|\leq |h(X,n-\hat{n})|+g_{S^2}(\bar{X},\bar{n}-\bar{\hat{n}}) \leq \\ \leq |h_{ij}||X^i||n^j-\hat{n}^j| + 2\sin^2(\theta) |\bar{n}^\phi-\bar{\hat{n}}^\phi|+2|\bar{n}^\theta-\bar{\hat{n}}^\theta|\leq \\ \leq \frac{6C\cosh(x)}{r}|n^t-\hat{n}^t|+\frac{6C}{r} \sum_{j=1}^3 |n^j-\hat{n}^j| +2 \sin^2(\theta) |\bar{n}^\phi-\bar{\hat{n}}^\phi|+2|\bar{n}^\theta-\bar{\hat{n}}^\theta|= \\=\frac{6C\cosh(x)}{r}|n^t-\hat{n}^t|+\frac{6C}{r} \sum_{j\neq \theta, \phi} |n^j-\hat{n}^j| +2\sin^2(\theta) |\bar{n}^\phi-\bar{\hat{n}}^\phi|+2|\bar{n}^\theta-\bar{\hat{n}}^\theta|.
	\end{multline}
	
	Using Lemma \ref{lem:angularvelocitiesfornulllines} we see that for any $\eps $ there exists $r$ such that $|\bar{n}|^2_{S^2},|\bar{\hat{n}}|^2_{S^2}<\eps^2$, i.e., in our charts $|n^\phi|<2 \eps,|\hat{n}^\phi|<2 \eps$ and $|n^\theta |<\eps,|\hat{n}^\theta|<\eps $. By Lemma \ref{lem:betaalpha} we estimate $|n^t|^2,|\hat{n}^t|^2<\frac{1+\eps^2}{\alpha_r \cosh^2(x)}$. Thus
	\[|-\cosh^2(x)X^t(n^t-\hat{n}^t)|\leq \frac{6C}{r}\big(6\eps+2\sqrt{\frac{1+\eps^2}{\alpha_r}}\big)+12\eps .\]

	Because $n$ is future pointing and $\hat{n}$ is past pointing we see that $|n^t-\hat{n}^t|=n^t+|\hat{n}^t|\geq \frac{2}{\sqrt{\beta_r}\cosh(x)}$, so
	$|X^t|\leq \frac{1}{\cosh(x)} \big(\frac{c}{r}+c' \eps \big)$ for some $c,c'$. 
	Thus for any $\eps >0$ there exists $r$ such that \beq \label{eq:Xtest} |X^t|\leq \frac{1}{\cosh(x)} \eps. \eeq
	
	Next we derive a similar estimate for $|X^x|$: Proceeding as before but looking at $\go(X,n+\hat{n})$ yields
	\beq |-\cosh^2(x)X^t(n^t+\hat{n}^t)+2X^x|\leq \big(\frac{c}{r}+c' \eps \big) \eeq
	for some $c,c'$. We now need to estimate $|n^t+\hat{n}^t|$. Since $\frac{1}{\sqrt{\beta_r}\cosh(x)}\leq |n^t|,|\hat{n}^t| \leq \frac{1+\eps}{\sqrt{\alpha_r}\cosh(x)}$ we have $|n^t+\hat{n}^t|=|n^t-|\hat{n}^t||\leq \frac{1}{\cosh(x)} (\frac{1+\eps}{\sqrt{\alpha_r}}-\frac{1}{\sqrt{\beta_r}})\leq \frac{1}{\cosh(x)} \eps$ for $r$ large. Combining this  with \eqref{eq:Xtest} we see $|\cosh^2(x)X^t(n^t+\hat{n}^t)|<\eps^2$ and hence for any $\eps>0$ we can find $r$ such that also \beq \label{eq:Xxest} |X^x|\leq \eps. \eeq
	Note that these estimates also ensure that $S$ is $\go$-spacelike.
	
	We now use this to estimate $K$: Let $p\in S$, let $X,Y$ be a $\go $-orthogonal basis for $T_pS$ with $|X^x|^2+|\bar{X}|^2_{S^2}=1$ and denote the Riemann tensor for $(S,g|_S)$ by $R_S$. Then \beq K(X,Y)=\frac{g(R_S(X,Y)Y,X)}{g(X,X)g(Y,Y)-g(X,Y)^2}=\frac{g(R(X,Y)Y,X)}{g(X,X)g(Y,Y)-g(X,Y)^2}\eeq 
	because $S$ is totally geodesic in $M$. 
	
	We start by estimating $K(X,Y)-\mathring{K}(X,Y)$ (where 

	$\mathring{K}(X,Y)=\frac{\go(\mathring{R}(X,Y)Y,X)}{\go(X,X)	\go(Y,Y)}$). First note that by \eqref{eq:Xtest} 
	\beq |\go(X,X)-1|=|\go(X,X)-|X^x|-|\bar{X}|^2_{S^2}|=\cosh^2(x
	)|X^t|^2<\eps\eeq and the same for $|\go(Y,Y)-1|$, i.e., $X$ and $Y$ are close to being $\go$-othonormal. From 
	$|\mathring{R}_{ijkl}-R_{ijkl}|\leq \frac{C\cosh^{\#t}(x)}{r}$ (see Remark \ref{rem:decayincoords}) and using \eqref{eq:Xtest},\eqref{eq:Xxest} for $X,Y$ and $|X^i|,|Y^i|\leq 2$ for $i=\theta,\phi$ (which follows from our choice of charts and  $\bar{X},\bar{Y}$ having unit $g_{S^2}$-norm) we see that \beq |\go(\mathring{R}(X,Y)Y,X)-g(R(X,Y)Y,X)|\leq \frac{c}{r} \eeq for some $c>0$. Similarly, using  $|h_{ij}|\leq \frac{C\cosh^{\#t}(x)}{r}$, we get  $|g(X,Y)|=|g(X,Y)-\go(X,Y)|<\frac{c}{r}$ and \beq 1+\frac{c}{r}+\eps \leq |g(X,X)|,|g(Y,Y)|\leq 1-\frac{c}{r}-\eps \eeq (note that $\go(X,X),\go(Y,Y)\in (1-\eps, 1+\eps)$). Putting these estimates together shows that indeed for any $\eps$ there exists $r$ such that 
	\[|K(X,Y)-\mathring{K}(X,Y)|\leq \eps\]
	as long as $S\sse M_1(r)\cup M_2(r)$.
	
	To estimate $\mathring{K}(X,Y)$ note that because $\go=g_{AdS_2}+g_{S^2}$, $K_{S^2}=1$ and \eqref{eq:Xtest},\eqref{eq:Xxest} we have \begin{multline}
	 |\go(\mathring{R}(\bar{X},\bar{Y})\bar{Y},\bar{X})-1|=|\go(\bar{X},\bar{X})\go(\bar{Y},\bar{Y})-\go(\bar{X},\bar{Y})^2-1|\leq \\ \leq |\go(X,X)\go(Y,Y)-\go(X,Y)^2-1|+c\eps. \end{multline}
	 So for any $\eps >0 $ we can find $r$ such that
	 \beq |\go(\mathring{R}(\bar{X},\bar{Y})\bar{Y},\bar{X})-1|<\eps. \eeq
	Finally,
	\beq 
	|\mathring{K}(X,Y)-1|= \big|\go(\mathring{R}(X,Y)Y,X)-1\big| \leq  \big|\go(\mathring{R}(\bar{X},\bar{Y})\bar{Y},\bar{X})-1\big| +c\eps < (c+1)\eps  
\eeq 
	and we are done.
\end{proof}

\begin{Theorem}\label{thm:round2spheres}
	The family $\{S_{u,v}\}_{(u,v)\in Q}$, where $Q =\{(u,v)\in \R^2 : u<v\;\mathrm{and}\; u_\infty >v_\infty \}$,
gives a continuous foliation of $M$ by totally geodesic round $2$-spheres.
\end{Theorem}
\begin{proof}
We start by	showing that $S_{u,v}$ is isometric to $S^2$. Let $r_0$ be large enough for Lemma~\ref{lem:noturningback} to apply. We will show that for any $(u,v)\in Q$ and $r>a$ there exists $u_0\equiv u_0(v,r)$ such that $[u,u_0]\times \{v\}\sse Q$ and $S_{u_0,v}\cap M_1(r_0)\sse M_1(r)$.  For any past null generator $\eta_u :(-\infty,x_p] \to M$ of $N_u$ starting in a point $p=(t_p,x_p,\omega_p)\in M_1(r_0)$ with $|x_p|=r$  we have $\eta_u(s)\in J^-((t_p+\tau_r,x_p, \omega(\eta_u(s)))$ by Lemma~\ref{lem:S2inIpaftertimetau}. So, since such a generator must be contained in $M_1(r)$, we get $\eta_u(s)\in I^-_{\go_{\beta_r}}((t_p+\tau_r,x_p, \omega(\eta_u(s)))$ by Lemma~\ref{lem:betaalpha}, and hence by Lemma~\ref{lem:nulllinseforgbeta} \beq t(\eta_u(s))<\frac{2}{\sqrt{\beta_r}}(\tan^{-1}(e^s)-\tan^{-1}(e^{-r}))+t_p+\tau_r \eeq
	if $\eta_u$ is parametrized with respect to the $x$-coordinate. Letting $s\to -\infty $ we get \beq t_p\geq u+\frac{2}{\sqrt{\beta_r}} \tan^{-1}(e^{-r})+\tau_r. \eeq
	A similar argument applied to $\hat{\eta}_v:(-\infty,x_p]\to M$, using that $\hat{\eta}_v(s)\in J^+((t_p-\tau_r,x_p, \omega(\hat{\eta}_v(s)))$, shows
	\beq t_p\leq v-\frac{2}{\sqrt{\beta_r}}\tan^{-1}(e^{-r})-\tau_r. \eeq
	So if $p\in S_{u,v}$, then \beq \frac{4}{\sqrt{\beta_r}}\tan^{-1}(e^{-r})+2\tau_r\leq v-u. \eeq
	Hence by choosing $u_0(v,r)<v$ as close to $v$ as necessary it follows that $S_{u_0,v}\cap M_1(r_0)\sse M_1(r)$. That $[u,u_0]\times \{v\}\sse Q$ is clear from $u\mapsto u_\infty $ being increasing, so $\bar{u}_\infty >u_\infty >v_\infty$ for all $\bar{u}>u$.
	
	Now connectedness of $S_{u_0(r),v}$ implies that even $S_{u_0(r),v}\sse M_1(r)$ for any $r>r_0$. Then by Prop.~\ref{prop:gaussCurvature} the Gauss curvature $K_{u_0(r),v}\to 1$ uniformly on $S_{u_0(r),v}$ as $r\to \infty$. But because all the $S_{u_0(r),v}$ are isometric to $S_{u,v}$, their Gauss curvatures (in corresponding points) have to be equal, so $K_{u,v}=1$. Together with $S_{u,v}$ being homeomorphic to $S^2$ this shows that $S_{u,v}$ is isometric to the round $2$-sphere.
	
	It remains to show that $\{S_{u,v}\}_{(u,v)\in Q }$ is a continuous foliation. This follows from the Frobenius theorem if we can show that $p\mapsto T_pS_{u_p,v_p}$ is a continuous distribution (it clearly is integrable, because it consists of tangent spaces to (smooth) submanifolds). Since $T_pS_{u_p,v_p}=\mathrm{span}\{n_p,\hat{n}_p\}^\perp $, where $n_p$ and $\hat{n}_p$ denote the future pointing null tangents (normed w.r.t.~some Riemannian background metric) to $N_{u_p}$ and $\hat{N}_{v_p}$ in $p$, it is sufficient to show continuity of $p\mapsto n_p$ and $p\mapsto \hat{n}_p$. Let $p_k\to p_0$ and let $\eta_n$ be the unique null geodesic generators of $N_{u_{p_k}}$ with $\dot{\eta}_k(0)=n_{p_k}$. Then the $\eta_k$ accumulate to a null line $\eta_0 $ passing through $p_0$ with $\eta \sse N_{u_{p_0}}$. Hence $n_{p_k}=\dot{\eta}_k(0)\to \dot{\eta}_0(0)=n_{p_0}$.
	\end{proof}

\section{Asymptotically $AdS_2\times S^2$ spacetimes with parallel Ricci tensor}
\label{sec:parallel}

In this section we will use the assumption of $\nabla \Ric =0$ to first obtain a general local splitting result, see Thm.~\ref{thm:localsplitting}, and finally a full rigidity result, see Thm.~\ref{thm:globalsplitting}. For $k>0$ we denote by $AdS_2(k)$ and $dS_2(k)$ two dimensional anti-de Sitter space with scalar curvature $-2k$ and  two dimensional de Sitter space with scalar curvature $2k$, respectively. Similarly $S^2(k)$ and $H^2(k)$ denote the two dimensional sphere with scalar curvature $2k$ and two dimensional hyperbolic space with scalar curvature $-2k$.

\begin{Theorem}\label{thm:localsplitting} Let $(M,g)$ be a (four dimensional, connected) spacetime with $\nabla \Ric=0$. If $R=0$ and $\Ric $ is non-degenerate, then there exists $k>0$ such that any $p\in M$ has a neighbourhood $U$ that is isometric to an open subset $V$ of $AdS_2(k)\times S^2(k)$ or $dS_2(k)\times H^2(k)$. 
\end{Theorem}

\begin{proof}
	First note that $\Ric $ cannot be proportional to the metric because $R=0$ but $\Ric\neq 0$ because it is non-degenerate. So \cite[Lemma~3.1]{Seno08} applies showing that for any open simply connected domain $(D,g)\sse (M,g)$ either the holonomy group is non-degenerately reducible or there exists
a covariantly constant null vector field $X$. But by the definition of $\Ric $ one clearly has $\Ric(X,Y)=0$ for any vector field $Y$ if $\nabla X=0$. So the existence of a covariantly constant vector field contradicts the  non-degeneracy of $\Ric$. Hence the holonomy group of $(D,g)$ is non-degenerately reducible. 
	
Now \cite[Prop.~3]{Wu} 
gives that any point $p$ in $M$ has a neighbourhood $U$ that is isometric to a direct product, say $U=L\times P$, where $L$ is Lorentzian and $P$ is Riemannian. First note that $\Ric_L$ and $\Ric_P$ are  non-degenerate (as bilinear forms  on $TL\times TL$, respectively $TP\times TP$): By the direct product structure $\Ric(X,Y)=0$ for $X\in TL$ and $Y\in TP$ so if $\Ric_L$ or $\Ric_P$ were degenerate, then so would be $\Ric $. Thus, $\dim(L)>1$ and $\dim(P)>1$, so the only possibility is $\dim(L)=2=\dim(P)$.
Neither $L$ nor $P$  splits, since any further splitting would give a one-dimensional factor contradicting non-degeneracy of $\Ric$. So both $\Ric_L$ and $\Ric_P$ have to be proportional to the respective metrics on $L$ and $P$, i.e.,  $\Ric_L=\lambda_L g_L$ and $\Ric_P=\lambda_P g_P$ with $\lambda_L+\lambda_P=0$. Setting $k:= |\lambda_P|=|\lambda_L|$, non-degeneracy of $\Ric$ implies $k>0$. So we have shown that for any $p\in M$ there exists a $k$ and a neighbourhood $U$ that is isometric to an open subset $V$ of $AdS_2(k)\times S^2(k)$ (if $\lambda_P>0$) or $dS_2(k) \times H^2(k)$ (if $\lambda_P<0$). Clearly $\lambda_P$, and thus $k$, is unique and locally constant, hence constant.
\end{proof}

\begin{remark}\label{rem:localsplitting}
	It is actually sufficient to assume that there exists a point $p_0$ such that $\Ric_{p_0}$ is non-degenerate and a sequence $p_n$ such that  $R_{p_n}\to 0$. This is obvious from the fact that $\nabla \Ric =0$ implies $\nabla R=0$, so $R=\mathrm{const.}$, and that if $\Ric_{p_0}(X_{p_0},.)=0$ then $\Ric(X,.)=0$ for any $X$ that is the parallel transport of $X_{p_0}$ along any curve.
\end{remark}

If $(M,g)$ is asymptotically $AdS_2\times S^2$, then $\lambda_P=1$ and the structure obtained in the previous section is consistent with this local product structure.

\begin{Corollary}\label{cor:localsplitting}
	Let $(M,g)$ be asymptotically $AdS_2\times S^2$ (in the sense of Def.~\ref{def:assumptions}) and assume that the null energy condition holds and that $\nabla \Ric =0$. Then any $p\in M$ has a neighbourhood $U$ that is isometric to an open subset $V\equiv L\times P$ of $AdS_2\times S^2$ (with metric $\go$). Further, the tangent space $T_qL$ is spanned by the vectors $n_q,\hat{n}_q$ and $T_qP=T_qS_{u_q,v_q}$ for all $q\in V$. 
\end{Corollary}

\begin{proof}
	Clearly $R_{p_n}\to 0$ as $x(p_n) \to \infty$ by the asymptotics (\ref{eq:curvatureasymptotics}). Also, there must exist a point $p$ where $\Ric_p$ is non-degenerate: Else we can find a sequence $p_n\in M_2$ with $x(p_n)\to \infty $ and vectors $X_n\in T_{p_n}M$ with $\Ric(X_n,.)=0$. We may assume that these $X_n$ are normed to $\cosh^2(x(p_n))|X_n^t|^2+|X_n^x|^2+|\bar{X}_n|^2_{S^2}=1$, so setting $Y_n:=X^t_n \partial_t-X^x_n \partial_x+\bar{X}_n$ we have $\mathring{\Ric}(X_n,Y_n)=1$ and $|\Ric (X_n,Y_n)-\mathring{\Ric}(X_n,Y_n)|\leq \frac{C}{|x(p_n)|}$. This contradicts $\Ric(X_n,.)=0$ for large enough $x(p_n)$.
	
	Thus, by Remark \ref{rem:localsplitting}, we can apply Theorem~\ref{thm:localsplitting}, to get $U\cong L\times P$. We have that $n,\hat{n}\sse TL$: If not, then $0=g_L(n,n)+g_P(n,n)$ and $g_P(n,n)\neq 0$, so $-g_L(n,n)=g_P(n,n)>0$ because $g_P$ is Riemannian. So $\Ric(n,n)=-k g_L(n,n)+kg_P(n,n)\neq 0$, contradicting $\Ric(n,n)=0$ (which follows from the NEC and $n,\hat{n}$ being tangent to null lines). Thus $T_qL$ is spanned by $n_q,\hat{n}_q$ and $T_qS_{u_q,v_q}=\mathrm{span}\{n_q,\hat{n}_q\}^\perp =T_qP$. Finally, because $S_{u,v}$ is isometric to the round $2$-sphere by Theorem~\ref{thm:round2spheres}, we must have $\lambda_P=1$, so $L\times P\sse AdS_2\times S^2$.
\end{proof}

Finally, the fact that the spheres $S_{u,v}$ are isometric to $S^2$ and hence geodesically complete allows us to globalize this splitting:

\begin{Theorem}\label{thm:globalsplitting}
	Let $(M,g)$ be asymptotically $AdS_2\times S^2$ (in the sense of Def.~\ref{def:assumptions}) and assume that the null energy condition holds and that $\nabla \Ric =0$.
 Then $M$ is isometric to $AdS_2\times S^2$.
\end{Theorem}

\begin{proof} From the local splitting in Cor.~\ref{cor:localsplitting} we see that the foliation $F:=\{S_{u,v}\}_{(u,v)\in Q}$ from Theorem~\ref{thm:round2spheres} must be smooth. Further, the distribution $q\mapsto \mathrm{span}\{n_q,\hat{n}_q\}= T_qS_{u_q,v_q}^\perp \sse T_qM$ must be smooth as well and hence by the Frobenius theorem give rise to a smooth foliation $K$ with leaves perpendicular to the leaves of $F$. Also,  from the local product structure we immediately see that both of these foliations are totally geodesic (i.e., their leaves are totally geodesic).

 For $F$ we know even more: Note that the leaves are exactly the spheres $S_{u,v}$ which are totally geodesic submanifolds isometric to $(S^2,d\Omega^2)$ by Theorem \ref{thm:round2spheres} and hence even geodesically complete. Finally, note the $M$ is simply connected because it is homeomorphic to $\R^2\times S^2$ (for any $x_0<-a$ the flow $\Phi^n:\R \times (\R\times \{x=x_0\}\times S^2)\to M$ of $n$ is a homeomorphism). 

So we may apply \cite[Cor.~2]{PR} to obtain that $M$ is globally isometric to a product $L\times P$ such that $K$ and $F$ correspond to the canonical foliations of the product $L\times P$. Since $P$ is a leaf of $F$, we see that $P=(S^2,d\Omega^2)$. And since $L$ is a leaf of $K$ it must be isometric to a non-empty open subset $U$ of $(AdS_2,g_{AdS_2})$. Further $L$ is null geodesically complete because the only null geodesics in $Q$ are null geodesic generators of the achronal null hypersurfaces $N_u$ and $\hat{N}_v$, hence complete by Lemma \ref{lem:maxnulliscomplete}. So all that remains is to show that any null geodesically complete non-empty open subset $U$ of $AdS_2$ must already be all of $AdS_2$: For any $p\in AdS_2\setminus U$ all null geodesics emanating from $p$ must also lie in $AdS_2\setminus U$. So if $U\neq AdS_2$ then $AdS_2\setminus U=AdS_2$ because any two points in $AdS_2$ can be connected by a curve consisting solely of null geodesic segments.\end{proof}

\appendix
\section{Asymptotics for the curvature}
In this appendix we give some details on the derivation of \eqref{eq:christofflasymptotics}, \eqref{eq:curvatureasymptotics} from \eqref{eq:decaycoordinates}, \eqref{eq:decay'coordinates}. Throughout this appendix we use $C$ to denote a running constant.

In general, if two pairs of functions $\mathring{f}_1,f_1$ and $\mathring{f}_2,f_2$ satisfy $|f_1-\fo_1 |\leq \frac{C}{|x|} |\fo_1(x)|$ and $|f_2-\fo_2 |\leq \frac{C}{|x|} |\fo_2(x)|$ on $\R \setminus [-a,a]$ then $|f_1|\leq C |\fo_1|$, $|f_2|\leq C|\fo_2|$ and  \beq \label{eq:generalproductest} |\fo_1 \fo_2 - f_1 f_2|\leq \frac{C}{|x|} |\fo_1(x)\fo_2(x)| \, \; \mathrm{on} \; \R\setminus [-a,a]. \eeq Using this, \eqref{eq:decaycoordinates} and the form of $\go $ (note that $\sin(\theta)$ is bounded away from zero in the charts we use) allows us to estimate
\begin{multline}
|\det(\go)-\det(g)|\leq |\det(\go)-\prod_{i=1}^4 g_{ii}|+\sum_{\sigma\neq \text{id}} \prod_{i=1}^4 |g_{i\sigma(i)}|\leq \\
\leq \frac{C}{|x|} \cosh^2(x) + \frac{C\cosh^2(x)}{|x|^2} (1+\frac{1}{|x|}+\frac{1}{|x|^2}) \leq \frac{C}{|x|} \cosh^2(x).
\end{multline}
From this we get
\beq \big|\frac{1}{\det(\go)}-\frac{1}{\det(g)}\big| \leq \frac{C}{|x|} \frac{1}{\cosh^2(x)}
\eeq
and using $A^{-1}=\frac{1}{det(A)} \text{adj}(A)$ this gives
\beq \label{eq:inversemetricasymptotics} |\go^{tt}-g^{tt}| \leq \frac{C}{|x|} \frac{1}{\cosh^2(x)},
\;\; |\go^{ti}-g^{ti}| \leq \frac{C}{|x|} \frac{1}{\cosh(x)} 
\;\; \text{and} \;\; |\go^{ij}-g^{ij}| \leq  \frac{C}{|x|}
\eeq
for $i,j\neq t$. Note that these imply
\beq \label{eq:inversemetricbounds} |g^{tt}| \leq  \frac{C}{\cosh^2(x)},
\;\; |g^{ti}| \leq \frac{C}{|x|} \frac{1}{\cosh(x)},
\;\;  |g^{ii}| \leq C\;\; \text{and} \;\;|g^{ij}| \leq \frac{C}{|x|} \; \mathrm{for}\; i\neq j.\eeq

Regarding the Christoffel symbols we note that 
\[\Gamma^l_{ij}=\frac{1}{2}\, g^{lk} \left( \partial_j g_{ki} + \partial_i g_{kj} - \partial_k g_{ij} \right).\]

Since only $\partial_x \go_{tt}$, $\partial_\theta \go_{\phi\phi}$ are non-zero, the estimates of all Christoffels not containing either of those derivatives follow from \eqref{eq:inversemetricbounds} and $|\partial_k g_{ij}|\leq \frac{C\cosh^{\#t}(x)}{|x|}$ for $(k,i,j)\neq (t,x,x),(\theta,\phi,\phi)$. So we have
\beq \label{eq:christoffelbounds} |\Gamma^k_{ij}| \leq  \frac{C}{|x|}\cosh^{\#t}(x)\eeq 
if $(k,i,j)\neq (t,t,x),(x,t,t),(\phi,\phi,\theta),(\theta,\phi,\phi)$.
The remaining Christoffels are $\Gamma^t_{tx}$, $\Gamma^x_{tt}$, $\Gamma^\phi_{\phi\theta}$ and $\Gamma^\theta_{\phi\phi}$. For these, the summands appearing in $|\mathring{\Gamma}-\Gamma |$ for which the $\mathring{g}$-part does not vanish can be estimated using $|\partial_x \go_{tt} -\partial_x g_{tt}|\leq \frac{C}{|x|}|\partial_x \go_{tt}|\leq \frac{C}{|x|}\cosh^2(x)$  (since $\cosh$ and $\sinh$ have the same behaviour at infinity) and $|\partial_\theta \go_{\phi\phi}-\partial_\theta g_{\phi\phi}|\leq \frac{C}{|x|}|\partial_\theta \go_{\phi\phi}|\leq \frac{C}{|x|}$ by \eqref{eq:decay'coordinates}, \eqref{eq:inversemetricasymptotics} and \eqref{eq:generalproductest}. This gives
\beq \label{eq:christofflbounds2} |\mathring{\Gamma}^k_{ij}-\Gamma^k_{ij}|\leq \frac{C}{|x|}\cosh^{\#t}(x)\;\; \mathrm{and}\;\;|\Gamma^k_{ij}|\leq C\cosh^{\#t}(x). \eeq for these four Christoffels.

For the components $R_{iklm}$ of the Riemann tensor we use
\beq R_{iklm}=\frac{1}{2}\left(
\partial_k\partial_l g_{im}
+ \partial_i\partial_m g_{kl}
- \partial_k\partial_m g_{il}
- \partial_i\partial_l g_{km} \right)
+g_{np} \left(
\Gamma^n{}_{kl} \Gamma^p{}_{im} - 
\Gamma^n{}_{km} \Gamma^p{}_{il} \right). \eeq
Again, if those products always contain at least one factor that is zero for $\go$, the desired estimates follows easily from the assumption on $\partial^2h$, $h$ and \eqref{eq:christoffelbounds},\eqref{eq:christofflbounds2}. The remaining two cases are $R_{xtxt}$ and $R_{\theta\phi\theta\phi}$ where $\go_{tt} \left(\mathring{\Gamma}^t_{tx} \mathring{\Gamma}^t_{xt} - 
\mathring{\Gamma}^t_{tt} \mathring{\Gamma}^t_{xx} \right)=\go_{tt} (\mathring{\Gamma}^t_{tx})^2 =\sinh^2(x)$ and $\go_{\phi\phi} \left(\mathring{\Gamma}^\phi_{\phi\theta} \mathring{\Gamma}^\phi_{\theta\phi} - \mathring{\Gamma}^\phi_{\phi\phi} \mathring{\Gamma}^\phi_{\theta\theta} \right)=\go_{\phi\phi} (\mathring{\Gamma}^\phi_{\phi\theta})^2=\cos^2(\theta)$, respectively. For these cases we again use \ref{eq:generalproductest} (and that $\sinh$ and $\cosh$ behave the same at infinity and that in our charts $\sin(\theta)$ is bounded away from zero).

Finally, the asymptotics for $\Ric$ and $R$ follow from \eqref{eq:inversemetricasymptotics},\eqref{eq:inversemetricbounds} and the asymptotics of $R_{iklm}$ using the same arguments.

\section{Weakening of the null energy condition}

In this appendix we wish to indicate how the results of this paper as summarized in Theorems \ref{thm:main} and \ref{thm:main2} continue to hold under the weaker integrated curvature condition \eqref{eq:averic}.  

\smallskip
The NEC enters into the proof of Theorem \ref{thm:main}  in two ways:
\ben
\item It is used in the proof of Lemma \ref{lem:angularvelocitiesfornulllines}.
\item It is used in results such as Theorem \ref{thm:existenceofNullHyp} which rely on the `null splitting theorem', Theorem IV.1 in \cite{G00}.
\een

The following is sufficient to ensure that Lemma \ref{lem:angularvelocitiesfornulllines} holds under the curvature condition \eqref{eq:averic}.

\begin{Proposition}
Assume $(M,g)$ satisfies the curvature condition \eqref{eq:averic}.  If $\eta:(-\infty,\infty) \to M$ is a complete null line then $\Ric(\eta'(s),\eta'(s)) = 0 $ for all $s \in \R$.
\end{Proposition}

\proof This follows almost immedetiately from Corollary 3.3 in \cite{EK94}.  Since $\eta$ is a complete null line, it is free of conjugate points. Then, by \cite[Corollary 3.3]{EK94}, 
$$
\int_{-\infty}^{\infty} \Ric(\eta'(s),\eta'(s)) ds \le 0 \,.
$$
But then the curvature condtion \eqref{eq:averic} implies that we have equality in the above.  In this case, \cite[Corollary 3.3]{EK94}  further implies that $\Ric(\eta'(s),\eta'(s)) = 0 $ for all $s \in \R$.\qed

\smallskip
The NEC enters into the proof of \cite[Theorem IV.1]{G00} in only one place, namely through Lemma IV.2.  The following proposition shows that this lemma remains valid under the curvature condition \eqref{eq:averic}. 

\begin{Proposition}\label{prop:support} 
Suppose $S$ is  an achronal $C^0$ future null hypersurface in $(M,g)$ whose null generators  are future geodesically complete.  If along each null generator $\eta: [0, \infty) \to \R$ the Ricci curvature satisfies \eqref{eq:averic} then $S$ has null mean curvature $\theta\ge 0$ in the sense of support hypersurfaces.\end{Proposition} 

We refer the interested reader to \cite{G00} for the definitions of terms being used in the statement of this proposition.  The proof makes use of the following lemma which is proved in \cite[Section 3]{Gal81}.

\begin{Lemma} Consider the intial value problem
\begin{align}\label{ivp}
x'' + p(s) x &= 0  \nonumber\\
x(0) &= 1 \\
x'(0) &= a \nonumber
\end{align}
If $p \in C^\infty([0,\infty))$ satisfies 
\beq
\int_0^{\infty} p(s) ds > a
\eeq
then the unique solution to \eqref{ivp} has a zero on $[0,\infty)$.
\end{Lemma}

\proof[Proof of Proposition \ref{prop:support}]  Given $p \in S$, let $\eta: [0, \infty) \to \R$ be a null generator of $S$ starting at $\eta(0) = p$.   For any $\e > 0$,    we have, 
\beq
\int_0^{\infty} \Ric(\eta'(s), \eta'(s)) ds  > -(n-2) \e \,. 
\eeq
By the lemma, the unique solution $x = x(s)$ to the initial value problem \eqref{ivp}, with 
\beq\label{data}
p(s) = \frac1{n-2} \Ric(\eta'(s),\eta'(s)) \quad \text{and}\quad a = -  \e
\eeq
satisfies 
$x(r_*) = 0$ for some  $r_* \in (0,\infty)$.  We may assume $r_*$ is the first zero of $x(s)$.

Fix $r > r_*$.  As in the proof of \cite[Lemma IV.2]{G00},  by considering $\d J^-(\eta(r))$ we obtain a smooth null hypersurface $S_r$  defined in a neighborhood of $\eta|_{[0,r)}$ such that $S_r$ is a past support hypersurface for $S$ at $p$.

Let $\theta = \theta(s)$ be the null expansion of $S_r$ along $\eta|_{[0,r)}$;
$\theta$ satisfies the Raychaudhuri equation \cite[(II.4)]{G00}.  Let y = y(s) be defined by the substitution,
$$
\frac{y'}{y} = \frac1{n-2} \theta(s)  \,
$$
with $y(0) = 1$.  A standard computation shows that $y$ satisfies the IVP \eqref{ivp} with
\beq
p(s) = \frac1{n-2}\left(\Ric(\eta',\eta') + \s^2 \right) \quad \text{and}\quad a =  \frac1{n-2} \theta(p)  \,.
\eeq

Suppose $\theta(0) < -(n-2) \e$.  By a basic ODE comparison result 
we have $y(s) \le x(s)$ (up to the first zero of $y$), where $x(s)$ is the solution to \eqref{ivp}$+$\eqref{data}. In particular $y(s)$ must go to zero somewhere on $[0, r_*]$.  This implies that  $\theta$ is not defined everywhere on this interval, which is a contradiction since $\theta= \theta(s)$ is smooth on $[0,r)$.  
Thus we must have $\theta(0) \ge -(n-2) \e$.  Since $\e$ is arbitrary, this proves the proposition.\qed

\smallskip
With regard to Theorem \ref{thm:main2}, the additional arguments of Section \ref{sec:parallel}, beyond those of  Section \ref{sec:main}, show that it is sufficient for the NEC, $\Ric(X,X) \ge 0$, to hold for vectors $X$ tangent to null rays.   But this follows trivially from \eqref{eq:averic}, since, under the assumption that $\Ric$ is covariant constant, the integrand is constant.

\bibliographystyle{amsplain}
\bibliography{ads2}

\end{document}